\documentclass[a4paper]{article}

\pdfoutput=1

\usepackage[T1]{fontenc}
\usepackage[utf8]{inputenc}
\usepackage[english]{babel}

\usepackage[margin=29mm]{geometry}

\usepackage{amsmath, amssymb, amsthm}
\usepackage{mathtools}
\usepackage{bm}
\usepackage{physics}
\usepackage{siunitx}

\usepackage[affil-it]{authblk}
\usepackage[font=footnotesize]{caption}
\usepackage{newfloat}
\usepackage{subcaption}
\usepackage{booktabs}
\usepackage{graphicx}
\usepackage{pgfplots}

\usepackage{csquotes}
\usepackage[style=phys,biblabel=brackets,chaptertitle=false,eprint,url]{biblatex}

\usepackage[colorlinks]{hyperref}
\usepackage[capitalise]{cleveref}

\pgfplotsset{compat=1.18}
\usepgfplotslibrary{external,fillbetween,groupplots}
\def\egap(#1,#2){sqrt(1-4*#2*(1-#2)*(#1-1)/#1)}
\def\schedprop(#1,#2){(1+(2*#2-1)/sqrt(1+4*(#1-1)*#2*(1-#2)))/2}
\def\schedorig(#1,#2){(1+tan((2*#2-1)*atan(sqrt(#1-1)))/sqrt(#1-1))/2}
\def\qofs(#1,#2){(1+(1-2*(1-#2)*(#1-1)/#1)/sqrt(1-4*#2*(1-#2)*(#1-1)/#1))/2}
\def\tauprop(#1,#2){#2-sqrt(#2*(1-#2)/(#1-1))}
\def\tauorig(#1,#2){1-acos(sqrt(#2))/acos(1/sqrt(#1))}
\def\qoftau(#1,#2){(1+2*(#2-1)*#1+sqrt(1+4*(#2-1)*#1*(1-#1)))/(2*#2)}  
\def\pofq(#1,#2){#2+#1^2*(1-2*#2)-2*sqrt(#1^2*(1-#1^2)*#2*(1-#2))}
\def\pofqsqrterr(#1,#2){(1-#1)^2*#2+2*#1*#2*(1-#1*#2-sqrt((1-#1^2*#2)*(1-#2)))}
\def\pofqlinerr(#1,#2){(1-#1)*#2+#1*#2*(1+#1*#2-2*#1*#2^2-2*sqrt((1-#1^2*#2^2)*#2*(1-#2)))}
\def\poftauconst(#1,#2,#3){((1+2*(#3-1)*#1+sqrt(1+4*(#3-1)*#1*(1-#1)))/(2*#3))+#2^2*(1-2*((1+2*(#3-1)*#1+sqrt(1+4*(#3-1)*#1*(1-#1)))/(2*#3)))-2*sqrt(#2^2*(1-#2^2)*((1+2*(#3-1)*#1+sqrt(1+4*(#3-1)*#1*(1-#1)))/(2*#3))*(1-((1+2*(#3-1)*#1+sqrt(1+4*(#3-1)*#1*(1-#1)))/(2*#3))))}  
\def\poftausqrt(#1,#2){((1+2*(#2-1)*#1+sqrt(1+4*(#2-1)*#1*(1-#1)))/(2*#2))+#1*(1-2*((1+2*(#2-1)*#1+sqrt(1+4*(#2-1)*#1*(1-#1)))/(2*#2)))-2*sqrt(#1*(1-#1)*((1+2*(#2-1)*#1+sqrt(1+4*(#2-1)*#1*(1-#1)))/(2*#2))*(1-((1+2*(#2-1)*#1+sqrt(1+4*(#2-1)*#1*(1-#1)))/(2*#2))))}  
\def\poftausinsqrt(#1,#2,#3){((1+2*(#3-1)*#1+sqrt(1+4*(#3-1)*#1*(1-#1)))/(2*#3))+#2^2*sin(sqrt(#1)/#2)^2*(1-2*((1+2*(#3-1)*#1+sqrt(1+4*(#3-1)*#1*(1-#1)))/(2*#3)))-2*sqrt(#2^2*sin(sqrt(#1)/#2)^2*(1-#2^2*sin(sqrt(#1)/#2)^2)*((1+2*(#3-1)*#1+sqrt(1+4*(#3-1)*#1*(1-#1)))/(2*#3))*(1-((1+2*(#3-1)*#1+sqrt(1+4*(#3-1)*#1*(1-#1)))/(2*#3))))}  
\def\poftausqrtscaled(#1,#2,#3){((1+2*(#3-1)*#1+sqrt(1+4*(#3-1)*#1*(1-#1)))/(2*#3))+#2^2*#1*(1-2*((1+2*(#3-1)*#1+sqrt(1+4*(#3-1)*#1*(1-#1)))/(2*#3)))-2*#2*sqrt(#1*(1-#2^2*#1)*((1+2*(#3-1)*#1+sqrt(1+4*(#3-1)*#1*(1-#1)))/(2*#3))*(1-((1+2*(#3-1)*#1+sqrt(1+4*(#3-1)*#1*(1-#1)))/(2*#3))))}  

\newtheorem{theorem}{Theorem}

\theoremstyle{definition}

\theoremstyle{remark}
\newtheorem*{remark}{Remark}

\newcommand{\bigO}{O}
\newcommand{\euler}{e}

\DeclareFloatingEnvironment{protocol}
\Crefname{protocol}{\protocolname}{\protocolname s}

\title{Adiabatic quantum unstructured search in parallel}

\author[ \hspace{-1ex}]{Sean A. Adamson\thanks{\href{mailto:sean.adamson@ed.ac.uk}{\texttt{sean.adamson@ed.ac.uk}}}}
\author[ \hspace{-1ex}]{Petros Wallden\thanks{\href{mailto:petros.wallden@ed.ac.uk}{\texttt{petros.wallden@ed.ac.uk}}}}
\affil[ \hspace{-1ex}]{School of Informatics, University of Edinburgh, \protect\\ 10 Crichton Street, Edinburgh EH8 9AB, United Kingdom}

\date{}

\begin{document}

\maketitle

\begin{abstract}
We present an optimized adiabatic quantum schedule for unstructured search building on the original approach of Roland and Cerf [\href{https://doi.org/10.1103/PhysRevA.65.042308}{Phys. Rev. A \textbf{65}, 042308 (2002)}]. Our schedule adiabatically varies the Hamiltonian even more rapidly at the endpoints of its evolution, preserving Grover's well-known quadratic quantum speedup. In the errorless adiabatic limit, the probability of successfully obtaining the marked state from a measurement increases directly proportional to time, suggesting efficient parallelization. Numerical simulations of an appropriate reduced two-dimensional Schr{\"o}dinger system confirm adiabaticity while demonstrating superior performance in terms of probability compared to existing adiabatic algorithms and Grover's algorithm, benefiting applications with possible premature termination. We introduce a protocol that ensures a marked-state probability at least $p$ in time of order $\sqrt{N}(1+p/\varepsilon)$, and analyze its implications for realistic bounded-resource scenarios. Our findings suggest that quantum advantage may still be achievable under constrained coherence times (where other algorithms fail), provided the hardware allows for them to be sufficiently long.
\end{abstract}

\section{Introduction}

The unstructured search problem is a fundamental computational task.
Given a database of $N$ elements with one marked item, to find this item classically $\bigO(N)$ queries to the database are required in the worst case.
Quantum algorithms offer a significant (although polynomial) advantage: a seminal algorithm due to \textcite{grover1996fast} achieves a reduction of time complexity to $\bigO(\sqrt{N})$ by discrete, repeated application of a unitary operator.
This quantum speedup is asymptotically tight in the circuit model \cite{bennett1997strengths}.
Adiabatic quantum computing provides an alternative continuous-time framework, where computation is realized through the gradual evolution of a quantum system under a time-dependent Hamiltonian \cite{farhi2000quantum}.
If the Hamiltonian is varied slowly enough, the system will remain close to its instantaneous ground state due to the adiabatic theorem \cite{born1928beweis}.
An adiabatic algorithm works by evolving the system from an initial Hamiltonian with a simple known ground state to a final Hamiltonian whose ground state encodes the solution to a problem, and then performing a measurement of this state.
The time complexity of such an algorithm is then determined by the amount of time spent performing this evolution adiabatically (assuming initial state preparation and final measurement times are negligible).
A fixed-rate linear interpolation between the simplest choices of initial and final Hamiltonians does not achieve any quantum advantage \cite{farhi2000quantum}.
However, the quadratic speedup of Grover's algorithm has been reproduced in the adiabatic setting by \textcite{roland2002quantum}.
This was achieved by varying the local speed at which the linear interpolation path is parametrized with respect to time, while still varying the Hamiltonian gradually enough at all times to maintain adiabaticity.

\subsection{Our contributions}

In this work, we introduce an optimized time schedule for adiabatic unstructured search, refining the well-known original interpolating schedule of \textcite{roland2002quantum}.
This schedule evolves the Hamiltonian even faster near the endpoints of the evolution, while slowing down in the critical middle region where the spectral energy gap is small.
The usual adiabatic condition is satisfied under our schedule (as it is this condition from which the schedule is derived).
While for general time-dependent Hamiltonians it is known that this condition does not necessarily imply adiabaticity in some situations \cite{marzlin2004inconsistency,tong2005quantitative,du2008experimental,wu2008adiabatic}, we later give numerical evidence that the adiabatic theorem does indeed hold in our case as expected.
We should stress here that despite being more aggressively scheduled, the running time for one full evolution of our adiabatic evolution is still $\bigO(\sqrt{N})$, matching Grover's algorithm.
As in circuit-based quantum computing, this is known to be optimal in the adiabatic model \cite{farhi1998analog,roland2002quantum}.

Our adiabatic schedule exhibits the interesting property that in the errorless ideal adiabatic limit (where we consider the system to always be exactly in its instantaneous ground state), the probability of measuring the system to be in the marked state increases in direct proportion with time, even at arbitrarily small times (see \cref{fig:q_tau}).
If such errorless evolution were possible in finite total time, it would immediately imply the possibility of efficient parallelization, in which $\sqrt{N}$ quantum processors could find the marked state in constant time (a known impossibility in the circuit model \cite{zalka1999grover}).

We relate the marked probability to the error (deviation from the instantaneous ground state due to finite evolution time) and analyze for which errors as functions of time parallelization is possible.
There exist nonzero error functions that would still admit parallelization.
We then perform numerical simulations to find the exact probabilities and errors for the system.
Perhaps unsurprisingly, these exact errors exceed the limits required for perfect parallelization.
However, we conjecture based on our numerical results a closed-form bound on the error (that of \cref{sec:sinsqrt_err}) that only fails to admit parallelization due to its behavior at times arbitrarily close to the evolution start.
On the other hand, we observe that the exact probability using our schedule outperforms the original schedule introduced by \textcite{roland2002quantum} for all times until approximately halfway through a whole evolution.
Furthermore, the probability outperforms Grover's algorithm after a constant time has passed, regardless of $N$.
This could find use in applications of unstructured search that operate with a chance of their search progress being canceled earlier than half way, and are otherwise likely to run until completion (such as certain cryptographic protocols).

Finally, we introduce \cref{prot:prob_p_measurement} that outputs the marked item with probability at least $p$ after time of order $\sqrt{N} + \bigO(\frac{1}{\varepsilon} \sqrt{N} p)$.
We analyze the overall running time in the more realistic scenario with bounded resources: a maximum available number of independent quantum processors each of which can run for a maximum time (for example, due to constrained coherence time of the hardware).
No quantum advantage is possible using Grover's algorithm or the original adiabatic schedule with constant bounded coherence time, however, our results give us reason to suspect that this is indeed possible using our schedule so long as coherence lasting longer than some constant time is achievable by the hardware.

\begin{protocol}[htb]
    \caption{
        Obtaining the marked state with success probability $p$ in quantum unstructured search within running time $\bigO(\sqrt{N} (1 + p / \varepsilon))$ for $p > 1/N$ and $\bigO(1)$ otherwise.
    }
    \label{prot:prob_p_measurement}
    \rule{\linewidth}{0.08em}

    Initial parameters $N = 2^{n} \geq 2$ where $n \in \mathbb{N}^{*}$, $0 < \varepsilon \ll 1$, and $0 \leq p \leq 1 - \varepsilon$.
    Problem Hamiltonian $H_{1} = I - \ketbra{m}$ whose ground state is the \emph{marked} $n$-qubit state $\ket{m}$ where $m \in \{0, \dots, N-1\}$.
    \begin{enumerate}
        \item Set the \emph{duration} parameter $T = \frac{2 \sqrt{N-1}}{\varepsilon}$ and the \emph{cutoff time} parameter
        \begin{equation}
            t_{\text{f}} = T(\varepsilon + p) = 2 \sqrt{N-1} \mathopen{}\left( 1 + \frac{p}{\varepsilon} \right) \mathclose{} .
        \end{equation}
        if $p > \frac{1}{N}$ and $t_{\text{f}} = 0$ otherwise.

        \item Prepare an initial state $\ket{\psi(0)} = \ket{+}^{\otimes n}$ (requiring a constant time).

        \item Evolve the system under Hamiltonian $[1-s] H_{0} + s H_{1}$ where $H_{0} = \ketbra{+}^{\otimes n}$ over time $t \in [0, t_{\text{f}}]$ according to the schedule
        \begin{equation}
            s = \frac{1}{2} \mathopen{}\left[ 1 + \frac{2 \frac{t}{T} - 1}{\sqrt{1 + 4 \mathopen{}\bigl( N-1 \bigr)\mathclose{} \frac{t}{T} \mathopen{}\bigl( 1 - \frac{t}{T} \bigr)\mathclose{}}} \right]\mathclose{} .
        \end{equation}

        \item Measure the resulting state $\ket{\psi(t_{\text{f}})}$ in the computational basis (requiring a constant time).
        The outcome will be $\ket{m}$ with probability at least $p$.
    \end{enumerate}

    \vspace*{-\baselineskip}
    \rule{\linewidth}{0.08em}
\end{protocol}

\subsection{Overview of techniques}

A detailed consideration of the transition matrix element for excitations between the ground and first-excited states allows us to even more aggressively vary the Hamiltonian at the beginning and end of its evolution, while still satisfying the adiabatic condition.
In order to give further evidence that evolution under our schedule is adiabatic (that the system remains close to the instantaneous ground state in the sense of \cref{thm:adiabatic_theorem}), as well as study the exact probability of successfully measuring the marked state at intermediate times, we turn to numerical methods to solve the Schr{\"o}dinger equation.
As we are most interested in the asymptotic behavior of the algorithms as $N = 2^{n}$ gets very large (where $n$ is the number of qubits encoding the number $N$), it is desirable for simulations to be computationally tractable for $n$ as large as possible.
We are mainly interested in simulating two quantities (as functions of time): the error and the probability.
To solve for only these, we are able to reduce the $2^{n}$-dimensional system of differential equations to a system of dimension just $2$ (details are contained in \cref{sec:numerical_schrodinger}), which can then be used to output both quantities simultaneously after some minor postprocessing.
Standard numerical software packages (described in \cref{sec:exact_adiabaticity}) can stably solve these equations, which we have tested up to $n=40$ (used to produce \cref{fig:crossing_time}).

\subsection{Related works}

The well-known statement of the adiabatic theorem that we make use of in this work has found criticism after certain pathological counterexamples and inconsistencies were discovered \cite{marzlin2004inconsistency,tong2005quantitative,du2008experimental,wu2008adiabatic}.
The first mathematically rigorous version is due to \textcite{kato1950adiabatic}.
Since then many more have been presented, each with their own assumptions and specific purposes \cite{jansen2007bounds,elgart2012note,ge2016rapid}.
Some algorithms designed for the noisy intermediate-scale quantum (NISQ) era in which we find ourselves have been inspired by adiabatic quantum computing \cite{garcia2018addressing,harwood2022improving,kolotouros2024simulating}.
Discretized versions of enhanced scheduling, such as introduced in the present work, may prove useful in potentially achieving quantum advantages in NISQ hardware implementing these algorithms.
More details on adiabatic quantum computing can be found in the review by \textcite{albash2018adiabatic}.

\subsection{Organization of the paper}

In \cref{sec:prelims}, we explain some notation, lay out the Schr{\"o}dinger equation for time-dependent Hamiltonians (making clear the different \emph{time} parameters involved), summarize the adiabatic theorem, and discuss properties of Hamiltonian for adiabatic unstructured search that we shall use throughout.
\Cref{sec:optimal_schedules} introduces our schedule, deriving it from the adiabatic condition.
\Cref{sec:truncated_evolution} analyzes truncated evolution (in which unstructured search algorithms may be terminated early with partial chance of successfully identifying the marked item) for both the errorless and approximate adiabatic regimes as well as for Grover's algorithm, which is the same as the errorless case of the original schedule.
In \cref{sec:const_err,sec:sqrt_err,sec:sinsqrt_err,sec:sqrt_scaled_err}, candidate error functions are presented, and their parallelization properties studied.
\Cref{sec:exact_adiabaticity} provides exact numerical simulations of errors and probabilities for the schedules, and a comparison to Grover's algorithm is analyzed.
Finally, \cref{sec:bounded_resources} explores resource-bounded scenarios, with constrained coherence times and limited parallel quantum processors.
We conclude in \cref{sec:discussion} with a summary of findings and potential directions for future work.

\section{Preliminaries}
\label{sec:prelims}

\subsection{Notation}

We will sometimes define some function $f$ taking values $f(s)$ written in terms of a parameter labeled $s$ (later called the \emph{scheduled time} or the \emph{linear interpolation} parameter).
Given also $s(\tau)$ defined in terms of a parameter $\tau$ (later the \emph{dimensionless time}), we will often write $f(\tau)$ in place of $f(s(\tau))$ in contexts where there is no ambiguity.
Where confusion could arise, we will use a tilde $\tilde{f}$ to denote the $\tau$-parametrized version of the function defined by
\begin{equation}
    \tilde{f}(\tau) = f(s(\tau)) .
\end{equation}

\subsection{The Schr{\"o}dinger equation and time-dependent Hamiltonians}

The time-dependent Schr{\"o}dinger equation on $t \in [0,T]$ with a time-dependent Hamiltonian is
\begin{equation}
\label{eq:schrodinger}
    \dv{t} \ket{\psi(t)} = -i \hat{H}(t) \ket{\psi(t)} .
\end{equation}
Parameterize by dimensionless time $\tau = t / T \in [0,1]$ so that
\begin{equation}
    \frac{1}{T} \dv{\tau} \ket*{\psi_{T}(\tau)} = -i H_{T}(\tau) \ket*{\psi_{T}(\tau)} .
\end{equation}
where we have defined
\begin{subequations}
\begin{align}
    \ket*{\psi_{T}(\tau)} & = \ket{\psi(T \tau)} , \\
    H_{T}(\tau) & = \hat{H}(T \tau) .
\end{align}
\end{subequations}
Finally, assume $H_{T}$ may only depend on $T$ via its argument, and so write $\tilde{H} \equiv H_{T}$.
Then
\begin{equation}
\label{eq:schrodinger_dimensionless}
    \dv{\tau} \ket{\psi_{T}(\tau)} = -i T \tilde{H}(\tau) \ket{\psi_{T}(\tau)} .
\end{equation}
This kind of indirect dependence on $T$ is the case, for example, for \emph{interpolating} Hamiltonians of the form
\begin{equation}
\label{eq:interp}
    \tilde{H}(\tau) = [1 - s(\tau)] H_{0} + s(\tau) H_{1} ,
\end{equation}
where $H_{0}$ and $H_{1}$ are some initial and final Hamiltonians, respectively, and the \emph{schedule} $s$ is strictly increasing with $s(0) = 0$ and $s(1) = 1$.
We will also assume that $s$ is differentiable.
Introducing notation for the special case of a linear interpolating Hamiltonian
\begin{equation}
\label{eq:linear_interp}
    H(s) = [1 - s] H_{0} + s H_{1} ,
\end{equation}
we can also write \cref{eq:interp} as
\begin{equation}
\label{eq:interp_as_linear}
    \tilde{H}(\tau) = H(s(\tau)) .
\end{equation}
To summarize the different \emph{time} parameters we keep track of:
\begin{enumerate}
    \item The \emph{duration} $T$ is the total time over which the Hamiltonian evolves from $H_{0}$ to $H_{1}$.
    It is a means to globally control the time scale over which the Hamiltonian varies.
    \item The \emph{physical} time $t \in [0,T]$ is the time according to real clocks, with the time evolution completing after $T$ physical time.
    \item The \emph{dimensionless} time $\tau = \frac{t}{T} \in [0,1]$ is the linear fraction of physical time that has elapsed of the full duration.
    \item The \emph{scheduled} time $s(\tau) \in [0,1]$ relates the passage of real time to the interpolation between initial and final Hamiltonians.
\end{enumerate}

\subsection{The adiabatic theorem}

Let $\ket{\varepsilon_{j}(s)}$ denote the instantaneous eigenstates satisfying
\begin{equation}
    H(s) \ket{\varepsilon_{j}(s)} = E_{j}(s) \ket{\varepsilon_{j}(s)} ,
\end{equation}
with corresponding energies $E_{0}(s) \leq E_{1}(s) \leq \dots$ for all $s \in [0,1]$.
Let the energy gap between the lowest two energy levels at $s$ be denoted
\begin{equation}
    g(s) = E_{1}(s) - E_{0}(s) .
\end{equation}
If we initially prepare the system in the instantaneous ground state $\ket{\varepsilon_{0}(0)}$ of $H(0)$ and adjust the Hamiltonian slowly enough, then the physical state of the system will remain close to the ground state of the evolving Hamiltonian.
This \emph{slowness} and \emph{closeness} is quantified by the adiabatic theorem.
\begin{theorem}[Adiabatic theorem]
    \label{thm:adiabatic_theorem}
    Let $0 < \varepsilon < 1$.
    If the \emph{adiabatic condition}
    \begin{equation}
    \label{eq:adiabatic_cond}
        \frac{2}{T} \frac{\left\lvert \! \mel{\varepsilon_{1}(\tau)}{\tilde{H}^{\prime}(\tau)}{\varepsilon_{0}(\tau)} \right\rvert}{g(\tau)^{2}} \leq \varepsilon
    \end{equation}
    is satisfied for all dimensionless times $\tau \in [0,1]$, then the ground-state probability at any $\tau$ is
    \begin{equation}
    \label{eq:adiabatic_error}
        \lvert \braket{\varepsilon_{0}(\tau)}{\psi_{T}(\tau)} \rvert^{2} \geq 1 - \epsilon(\tau)^{2}
    \end{equation}
    for some error function $\epsilon \colon [0,1] \to [0,\varepsilon]$.
\end{theorem}
\begin{remark}
    Although called a ``theorem'', sufficient conditions on the Hamiltonian for this adiabatic theorem to hold are currently unknown and various pathological counterexamples have been discovered \cite{marzlin2004inconsistency,tong2005quantitative,du2008experimental,wu2008adiabatic}.
    Nevertheless, we need only that it holds in the special case of our unstructured search Hamiltonian $\tilde{H}$, which is defined in \cref{sec:search_hamiltonian}.
\end{remark}

The resulting closeness in \cref{thm:adiabatic_theorem} includes that at the final physical time $T$ (when $\tau = 1$), which is usually the time of interest in adiabatic quantum computing.
Under the adiabatic theorem, the square of the (typically small) diabaticity parameter $\varepsilon \ll 1$ has the interpretation of a maximum probability with which at any particular time the system would be found to have been excited from its instantaneous ground state if measured.
The error function appearing in \cref{eq:adiabatic_error} is usually taken to be the special case of the largest possible constant $\epsilon(\tau) = \varepsilon$.
Here, we have allowed other functions (for which \cref{eq:adiabatic_error} may or may not follow from the assumptions of the theorem).
We will consider some possible choices for these and their diabaticity in \cref{sec:approx_adiabatic}.

Note that we choose to include a factor of $2$ in \cref{eq:adiabatic_cond} compared to some other statements of the adiabatic theorem, which we could also absorb into $\varepsilon$ on the right side.
Our version of the theorem that \cref{eq:adiabatic_cond} implies \cref{eq:adiabatic_error} is weaker than the usual statement in the sense that it automatically holds if the theorem is true when stated without the factor of $2$.

\subsection{Hamiltonian for unstructured search}
\label{sec:search_hamiltonian}

In the unstructured search problem over $N = 2^{n}$ strings in $\{0,1\}^{n}$, the final Hamiltonian takes the form
\begin{equation}
\label{eq:H1}
    H_{1} = I - \ketbra{m} ,
\end{equation}
where $m \in \{0,1\}^{n}$ and $\ket{m}$ is known as the marked state.
While there are many different possible choices for the initial Hamiltonian, the most mathematically convenient is
\begin{equation}
\label{eq:H0}
    H_{0} = I - \ketbra{\phi} ,
\end{equation}
where $\ket{\phi} = \frac{1}{\sqrt{N}} \sum_{j=0}^{N-1} \ket{j} = \ket{+}^{\otimes n}$ is the $n$-qubit uniform superposition state.
The ground states of $H_{0}$ and $H_{1}$ are $\ket{m}$ and $\ket{\phi}$ respectively, with corresponding energies of $0$ (this is by convention and is the reason for the identity operators appearing in \cref{eq:H0,eq:H1}).
If $H$ interpolates these $H_{0}$ and $H_{1}$ as in \cref{eq:linear_interp} then its two lowest instantaneous energy levels are
\begin{subequations}
\begin{align}
    E_{0}(s) & = \frac{1}{2} [1 - g(s)] , \\
    E_{1}(s) & = \frac{1}{2} [1 + g(s)] ,
\end{align}
\end{subequations}
where the energy gap $g(s)$ is
\begin{equation}
    g(s) = \sqrt{1 - 4 \frac{N-1}{N} s(1-s)} .
\end{equation}
The minimum energy gap, occurring at $s = 1/2$, is
\begin{equation}
    g_{\text{min}} = \frac{1}{\sqrt{N}} .
\end{equation}

\begin{figure}
    \centering
    \includegraphics{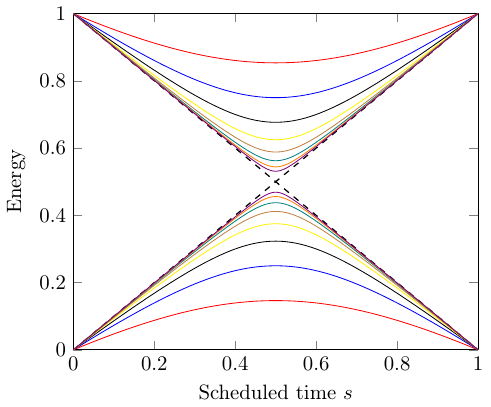}
    \caption{
        Ground and first-excited energy levels for the linear interpolating Hamiltonian defined via \cref{eq:linear_interp} by the initial and final Hamiltonians of \cref{eq:H0,eq:H1}.
        The levels shown are for search domains of sizes $N = 2^{n}$ where $n \in \{1, \dots, 8\}$.
        Dashed lines represent the asymptotic limits of the two levels as $N \to \infty$.
    }
    \label{fig:energy_levels}
\end{figure}

Now define the uniform superposition over all unmarked basis states
\begin{equation}
\label{eq:marked_perp}
    \ket{m^{\perp}} = \frac{1}{\sqrt{N-1}} \sum_{j \neq m} \ket{j} = \sqrt{\frac{N}{N-1}} \ket{\phi} - \frac{1}{\sqrt{N-1}} \ket{m} .
\end{equation}
The two lowest energy states of $H(s)$ are
\begin{subequations}
\label{eq:gs_vectors}
\begin{align}
    \label{eq:gs_vector_0}
    \ket{\varepsilon_{0}(s)} & = \alpha(s) \ket{m} + \beta(s) \ket{m^{\perp}} , \\
    \ket{\varepsilon_{1}(s)} & = - \beta(s) \ket{m} + \alpha(s) \ket{m^{\perp}} ,
\end{align}
\end{subequations}
where the coefficients $\alpha$ and $\beta$ are given by
\begin{subequations}
\label{eq:gs_coefficients}
\begin{align}
    \label{eq:gs_coefficient_marked}
    \alpha(s) & = \sqrt{\frac{1 + c(s)}{2}} , \\
    \beta(s) & = \sqrt{\frac{1 - c(s)}{2}} ,
\end{align}
\end{subequations}
with $- \frac{N-2}{N} \leq c(s) \leq 1$ and defined by
\begin{equation}
\label{eq:cos}
    c(s) = \frac{1}{g(s)} \mathopen{}\left[ 1 - 2 \frac{N-1}{N} (1-s) \right]\mathclose{} .
\end{equation}

\section{Optimal schedules for unstructured search}
\label{sec:optimal_schedules}

Explicit calculation (see \cref{sec:adiabatic_numerator}) shows
\begin{equation}
\label{eq:matrix_element}
    \left\lvert \! \mel{\varepsilon_{1}(s)}{H^{\prime}(s)}{\varepsilon_{0}(s)} \right\rvert
    = \frac{\sqrt{N-1}}{N} \frac{1}{g(s)}
    \leq \frac{\sqrt{N-1}}{N} \frac{1}{g_{\text{min}}}
    = \sqrt{1 - \frac{1}{N}}
    < 1 .
\end{equation}
Hence, according to the adiabatic condition of \cref{eq:adiabatic_cond} using a linear schedule $s(\tau) = \tau$, it is seen that taking
\begin{equation}
    T \geq \frac{2}{\varepsilon g_{\text{min}}^{2}} = \frac{2N}{\varepsilon}
\end{equation}
results in sufficiently slow variation of the Hamiltonian to satisfy the adiabatic condition.
Note then that by choosing a linear schedule we have not found any speedup over the classical time complexity of unstructured search.

Due to \cref{eq:interp_as_linear} we have
\begin{equation}
    \dv{\tilde{H}}{\tau} = \dv{s}{\tau} \dv{H}{s} .
\end{equation}
Inserting this into \cref{eq:adiabatic_cond} gives
\begin{equation}
    \frac{\left\lvert \! \mel{\varepsilon_{1}(s(\tau))}{H^{\prime}(s(\tau))}{\varepsilon_{0}(s(\tau))} \right\rvert}{g(s(\tau))^{2}} \frac{2}{T} \dv{s}{\tau} \leq \varepsilon .
\end{equation}
Treating $\tau$ as the inverse to $s$ and rearranging gives
\begin{equation}
\label{eq:adiabatic_ode}
    \dv{t}{s} = T \dv{\tau}{s} \geq \frac{2}{\varepsilon} \frac{\left\lvert \! \mel{\varepsilon_{1}(s)}{H^{\prime}(s)}{\varepsilon_{0}(s)} \right\rvert}{g(s)^{2}} .
\end{equation}
\Textcite{roland2002quantum} used the adiabatic condition in the form of \cref{eq:adiabatic_ode} to explicitly extract an optimal schedule and recover the same quadratic speedup provided by Grover's algorithm after evolving to the final time ($\tau = 1$).
They bounded the numerator in \cref{eq:adiabatic_ode} above by a constant, as is shown possible by \cref{eq:matrix_element}.
Here we instead use its exact expression.

Rewriting the right-hand side of \cref{eq:adiabatic_ode} explicitly using the first equality in \cref{eq:matrix_element} and assuming saturation of the inequality in the adiabatic condition for optimality gives
\begin{equation}
    T \dv{\tau}{s}
    = \frac{\sqrt{N-1}}{\varepsilon N} \frac{2}{g(s)^{3}} .
\end{equation}
We now solve for $\tau$.
Integrating with the boundary condition $s(0) = 0$ gives
\begin{equation}
\label{eq:integrate_tau}
\begin{split}
    t = T \tau
    & = \frac{2 \sqrt{N-1}}{\varepsilon N} \int_{0}^{s} \frac{1}{g(x)^{3}} \dd{x} \\
    & = \frac{2 \sqrt{N-1}}{\varepsilon N} \int_{0}^{s} \left[ 1 - 4 \frac{N-1}{N} x(1-x) \right]^{- \frac{3}{2}} \dd{x} \\
    & = \frac{\sqrt{N-1}}{\varepsilon} \mathopen{}\left[ 1 + \frac{2s-1}{\sqrt{1 - 4 \frac{N-1}{N} s(1-s)}} \right]\mathclose{} .
\end{split}
\end{equation}
Applying the other boundary condition $s(1) = 1$ gives the relation between the $\varepsilon$ appearing in \cref{eq:adiabatic_cond} and the final physical time
\begin{equation}
\label{eq:final_time}
    T = \frac{2 \sqrt{N-1}}{\varepsilon} ,
\end{equation}
which also shows that we recover the same Grover speedup as in \cite{roland2002quantum} after the full evolution has completed.
Combining \cref{eq:integrate_tau,eq:final_time}, the dimensionless time corresponding to scheduled time $s$ is
\begin{equation}
\label{eq:tau_at_s}
    \tau = \frac{1}{2} \mathopen{}\left[ 1 + \frac{2s-1}{\sqrt{1 - 4 \frac{N-1}{N} s(1-s)}} \right]\mathclose{} .
\end{equation}
This can also be inverted to find the optimal schedule
\begin{equation}
\label{eq:optimal_schedule}
    s(\tau) = \frac{1}{2} \mathopen{}\left[ 1 + \frac{2 \tau - 1}{\sqrt{1 + 4 (N-1) \tau (1 - \tau)}} \right]\mathclose{} .
\end{equation}
See \cref{tab:schedule_comparison} for a comparison with the original optimal schedule of \textcite{roland2002quantum}.

\begin{table}
    \centering
    \caption{
        Comparison between our proposed schedule and the original optimal schedule of \textcite{roland2002quantum} for the evolution duration $T$ (in terms of $N$ and $\varepsilon$) and relations between the dimensionless time $\tau$ and scheduled time $s$ (which are inverse of each other).
    }
    \label{tab:schedule_comparison}
    \begin{tabular}{lccc}
        \toprule
         Schedule & Duration parameter $T$ & Dimensionless time $\tau$ & Scheduled time $s$ \\
         \midrule
         This work & $\frac{2}{\varepsilon} \sqrt{N-1}$ & $\frac{1}{2} \mathopen{}\left[ 1 + \frac{2s-1}{\sqrt{1 - 4 \frac{N-1}{N} s(1-s)}} \right]\mathclose{}$ & $\frac{1}{2} \mathopen{}\left[ 1 + \frac{2 \tau - 1}{\sqrt{1 + 4 (N-1) \tau (1 - \tau)}} \right]\mathclose{}$ \\
         \addlinespace
         Original \cite{roland2002quantum} & $\frac{2}{\varepsilon} \frac{N \arctan{\sqrt{N-1}}}{\sqrt{N-1}}$ & $\frac{1}{2} \mathopen{}\left[ 1 + \frac{\arctan[(2s-1) \sqrt{N-1}]}{\arctan{\sqrt{N-1}}} \right]\mathclose{}$ & $\frac{1}{2} \mathopen{}\left[ 1 + \frac{\tan[(2 \tau - 1) \arctan{\sqrt{N-1}}]}{\sqrt{N-1}} \right]\mathclose{}$ \\
         \bottomrule
    \end{tabular}
\end{table}

As can be seen from the example in \cref{fig:schedule_comparison}, the proposed schedule differs qualitatively from the original in that we increase the scheduled time parameter $s$ even more rapidly with respect to physical time around the start and end of evolution than in the original, while still satisfying the adiabatic condition of \cref{eq:adiabatic_cond}.

\begin{figure}
    \centering
    \includegraphics{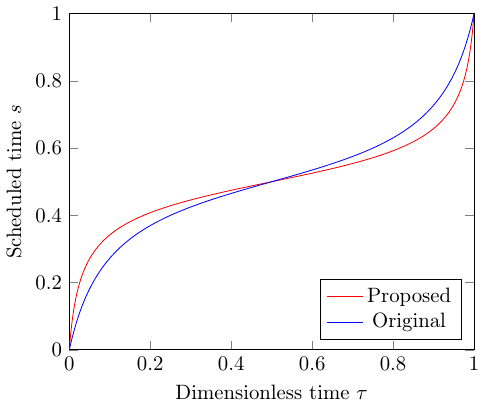}
    \caption{
        Example schedule functions of search domain size $N=16$ for our proposed optimal schedule \cref{eq:optimal_schedule} and the original optimal schedule of \textcite{roland2002quantum}.
        Both yield a $T = \bigO(\sqrt{N})$ time complexity for the final time at $\tau=1$ (see \cref{tab:schedule_comparison}), however, the proposed schedule allows us to increase the scheduled time $s$ faster with respect to physical (dimensionless shown here) time $\tau$ at the start and end of evolution.
    }
    \label{fig:schedule_comparison}
\end{figure}

\section{Truncated adiabatic evolution}
\label{sec:truncated_evolution}

We now consider a truncated evolution, in which we may stop varying the Hamiltonian at some \emph{cutoff time} before the final Hamiltonian $H_{1}$ has been reached.
That is, we may choose to measure the state of the system when $\tau < 1$ and observe the marked state with some probability.

\subsection{Evolution of the ideal ground state}

Let us first consider an (unphysical) ideal evolution in which the initial state remains in the exact instantaneous ground state of our Hamiltonian at all times.
At a scheduled time of $s \in [0,1]$ in our interpolating Hamiltonian \cref{eq:linear_interp}, the probability $q \in \left[ \frac{1}{N}, 1 \right]$ with which the marked state would occur if a measurement were performed on its instantaneous ground state is found from \cref{eq:gs_vector_0} to be
\begin{equation}
\label{eq:marked_prob}
    q = \lvert \alpha(s) \rvert^{2}
    = \frac{1 + c(s)}{2}
    = \frac{1}{2} \mathopen{}\left[ 1 + \frac{1 - 2 \frac{N-1}{N} (1-s)}{\sqrt{1 - 4 \frac{N-1}{N} s(1-s)}} \right]\mathclose{} .
\end{equation}
This function is depicted in \cref{fig:s_q}.
It is also possible to invert \cref{eq:marked_prob} to instead find the a general scheduled time in terms of the marked probability
\begin{equation}
\label{eq:s_of_p_ideal}
    s =
    \begin{cases}
        \frac{1}{2} \mathopen{}\left[ 1 - \frac{\frac{N}{\sqrt{N-1}} 2q(1-2q) \sqrt{\frac{1-q}{q}} + 1}{4 N q(1-q) - 1} \right]\mathclose{} & \text{if $q \neq \frac{1}{2} \mathopen{}\left[ 1 + \sqrt{\frac{N-1}{N}} \right]\mathclose{}$,} \\
        \frac{1}{4} \mathopen{}\left( 3 - \frac{1}{N-1} \right)\mathclose{} & \text{if $q = \frac{1}{2} \mathopen{}\left[ 1 + \sqrt{\frac{N-1}{N}} \right]\mathclose{}$.}
    \end{cases}
\end{equation}

\begin{figure}
    \centering
    \includegraphics{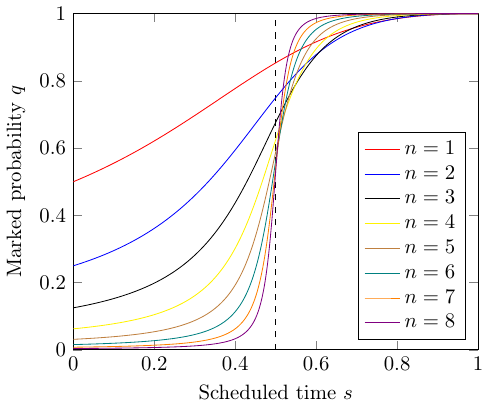}
    \caption{
        The probability with which a measurement of the instantaneous ground state of the linear interpolating Hamiltonian \cref{eq:linear_interp} defined by initial and final Hamiltonians \cref{eq:H0,eq:H1} will result in the marked state $\ket{m}$ as outcome.
        Different lines correspond to different search domain sizes $N = 2^{n}$.
    }
    \label{fig:s_q}
\end{figure}

Let us rewrite the ideal probability $q$ in terms of dimensionless time $\tau$ in the cases of both schedules that we consider.
We do this by inserting each schedule (see \cref{tab:schedule_comparison}) into \cref{eq:marked_prob} as $s$.

\paragraph{Original schedule}

Choosing the original schedule results in (see \cref{sec:orig_sched_prob})
\begin{equation}
\label{eq:original_prob_tau}
    q = \cos^{2} \mathopen{}\left[ (1 - \tau) \arccos{\frac{1}{\sqrt{N}}} \right]\mathclose{} .
\end{equation}
\Cref{eq:original_prob_tau} can also be inverted on the domain $q \in \left[ \frac{1}{N}, 1 \right]$ as
\begin{equation}
    \tau = 1 - \frac{\arccos{\sqrt{q}}}{\arccos{\frac{1}{\sqrt{N}}}} .
\end{equation}
Finally, as $N \to \infty$ we get
\begin{subequations}
\label{eq:grover_asymptotic}
\begin{align}
    q & \to \sin^{2} \mathopen{}\left( \frac{\pi}{2} \tau \right)\mathclose{} , \\
    \label{eq:asymptotic_tau_grover}
    \tau & \to \frac{2}{\pi} \arcsin{\sqrt{q}} .
\end{align}
\end{subequations}

It is interesting to note that \cref{eq:original_prob_tau} and the subsequent expressions are the exact same as are found for Grover's algorithm if $\tau$ instead represents the fraction of Grover iterations that have elapsed out of the total number required to reach unit probability (see \cref{sec:grover_probability}).
With every Grover iteration being assumed to take an identical and constant amount of physical time, $\tau$ represents the same dimensionless time quantity in both adiabatic algorithms and Grover's algorithm.
This implies that, rather than considering the adiabatic condition \cref{eq:adiabatic_cond} to derive the original schedule of \textcite{roland2002quantum}, one could instead reverse engineer it from the probability given by Grover's algorithm and then observe its adiabaticity a posteriori.

\begin{theorem}
\label{thm:grover_equiv}
    In both Grover's algorithm and the original adiabatic unstructured search algorithm of \textcite{roland2002quantum} (in its errorless form), if the state of their systems are measured in the computational basis at any dimensionless time $\tau \in [0,1]$, the probability that the outcome is the marked state is given by \cref{eq:original_prob_tau} to be
    \begin{equation}
        \cos^{2} \mathopen{}\left[ (1 - \tau) \arccos{\frac{1}{\sqrt{N}}} \right]\mathclose{} ,
    \end{equation}
    where $N$ is the number of elements in the search domain.
\end{theorem}
\begin{proof}
    See \cref{sec:grover_probability} for Grover's algorithm and \cref{sec:orig_sched_prob} for the adiabatic case with its original interpolation schedule.
\end{proof}

\paragraph{Proposed schedule}

Choosing our schedule \cref{eq:optimal_schedule} results in
\begin{equation}
\label{eq:q_of_tau_proposed}
    q = \frac{1}{2N} \mathopen{}\left[ 1 + 2(N-1)\tau + \sqrt{1 + 4(N-1) \tau (1-\tau)} \right]\mathclose{} .
\end{equation}
Inverting this yields
\begin{equation}
\label{eq:tau_of_q}
    \tau = q - \sqrt{\frac{q(1-q)}{N-1}} \leq \frac{N}{N-1} q - \frac{1}{N-1} \leq q
\end{equation}
on the domain $q \in [\frac{1}{N}, 1]$.
See \cref{fig:q_tau} for a visual depiction of this function, as well as the equivalent found using the original schedule of \textcite{roland2002quantum} (and Grover's algorithm).

\begin{figure}
    \centering
    \includegraphics{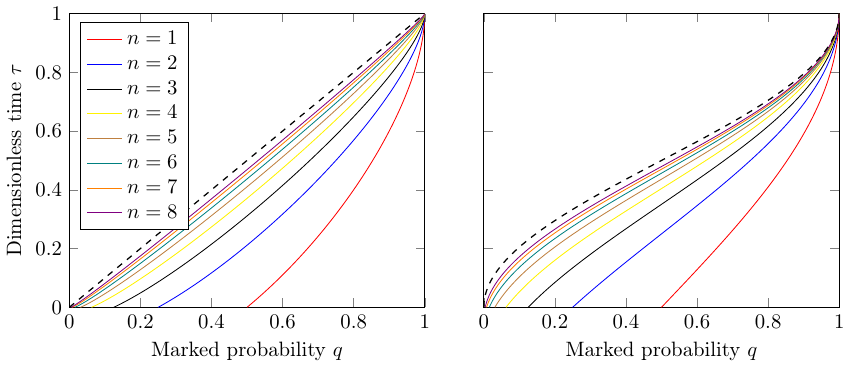}
    \caption{
        Dimensionless time $\tau$ of evolution that must pass before probability $q$ that a measurement of the instantaneous ground state of \cref{eq:interp} in the computational basis would yield the marked state $\ket{m}$, depending on the schedule used.
        The \textbf{left} plot corresponds to our proposed schedule in \cref{eq:optimal_schedule} while the \textbf{right} plot corresponds to the original optimal schedule of \textcite{roland2002quantum} (and also to Grover's algorithm).
        Search domain sizes shown are $N = 2^{n}$.
        The asymptotic upper limits as $N \to \infty$ are depicted by dashed lines ($\tau \to q$ and \cref{eq:asymptotic_tau_grover}, respectively).
        Note that in the original (right-side) case $\tau$ cannot be bounded above for all $N$ by any function directly proportional to $q$ due to its diverging gradient at $q=1/N$ as $N \to \infty$.
    }
    \label{fig:q_tau}
\end{figure}

Combining \cref{eq:tau_of_q} with \cref{eq:final_time}, we see that under ideal evolution (exactly tracking the instantaneous ground state of our Hamiltonian), we would achieve a marked state probability of $q$ in physical time
\begin{equation}
\label{eq:tq_linear}
    t = T \tau \leq \frac{2 \sqrt{N}}{\varepsilon} q .
\end{equation}
\begin{theorem}
\label{thm:ideal_parallelization}
    In the ideal (unphysical) adiabatic regime of vanishing error $\epsilon(\tau) = 0$, at most $S = \lceil N^{\eta} \rceil$ independent quantum processors running in parallel are required to perform unstructured search in $\bigO(N^{\frac{1}{2} - \eta})$ time for any $\eta \geq 0$.
\end{theorem}
\begin{proof}
    Choose $q = 1/S$ for each quantum processor and perform their evolution in parallel.
    Measuring the state of each processor after physical time $t = \frac{2\sqrt{N}}{\varepsilon S} = \frac{2\sqrt{N}}{\varepsilon \lceil N^{\eta} \rceil} = \bigO(N^{\frac{1}{2} - \eta})$ has passed, we expect exactly $1$ of the outcomes to have been the marked state.
    The probability that after measuring, none of the outcomes were marked is $(1 - 1/S)^{S} \leq \euler^{-1}$.
    The protocol of preparing, evolving, and measuring may be repeated $R \geq 1$ times to achieve some choice of arbitrarily small probability at most $\euler^{-R}$ that no outcome is marked.
\end{proof}

In other words, in this ideal (and unphysical) adiabatic regime, it would be possible to perfectly trade $\bigO(\sqrt{N})$ time on $\bigO(1)$ quantum processors for $\bigO(1)$ time on $\bigO(\sqrt{N})$ quantum processors.
On the other hand, \cref{eq:grover_asymptotic} shows that for the original optimal schedule of \textcite{roland2002quantum} and Grover's algorithm, it is not possible to bound $\tau$ above for all $N$ by any function directly proportional to $q$ as in \cref{eq:tq_linear}.
Therefore, it is not possible to efficiently parallelize quantum unstructured search even in the ideal adiabatic regime using the original adiabatic schedule or Grover's algorithm \cite{zalka1999grover}.

\subsection{Approximately adiabatic regime}
\label{sec:approx_adiabatic}

Thus far we have only made use of the adiabatic condition of \cref{eq:adiabatic_cond} (to derive a satisfactory candidate schedule), and have not considered the deviation of the system from the instantaneous ground state of our Hamiltonian over the course of its evolution.
This is important because measurement of the physical state may yield the marked state with a lower probability than measurement of the instantaneous ground state of the Hamiltonian (something that we cannot physically perform).

Consider initially preparing the ground state $\ket{\psi(0)} = \ket{\varepsilon_{0}(0)}$ and then evolving according to our schedule of \cref{eq:optimal_schedule} until a physical time $t$ (equivalently dimensionless time $\tau$).
If the ideal ground state $\ket{\varepsilon_{0}(\tau)}$ were to be measured in the computational basis then the marked state would occur with a probability $q$ given by \cref{eq:q_of_tau_proposed}.
If the physical state $\ket{\psi(\tau)}$ were measured, let us call the probability of finding the marked state $p$.
This is the probability of interest to us, as it is what would actually be observed in an experiment.
Due to the adiabatic theorem (\cref{thm:adiabatic_theorem}) the state satisfies \cref{eq:adiabatic_error}.
We may further assume without loss of generality that $\ket{\psi(\tau)}$ takes the form of a linear combination of the marked basis state and the uniform superposition of all unmarked basis states.
This is because among states of this form are those that give the lowest possible value for $p$ while still satisfying \cref{eq:adiabatic_error}.
In this case, \cref{eq:adiabatic_error} becomes
\begin{equation}
\label{eq:overlap_in_pq}
    pq + (1-p)(1-q) + 2 \sqrt{pq(1-p)(1-q)} \geq 1 - \epsilon(\tau)^{2} .
\end{equation}
Assuming the worst case that $p \leq q$, \cref{eq:overlap_in_pq} implies
\begin{equation}
\label{eq:q_of_p_inequality}
    q \leq
    \begin{cases}
        p + \epsilon^{2} (1-2p) + 2 \sqrt{\epsilon^{2} (1 - \epsilon^{2}) p(1-p)} & \text{if $p < 1 - \epsilon(\tau)^{2}$,} \\
        1 & \text{if $p \geq 1 - \epsilon(\tau)^{2}$.}
    \end{cases}
\end{equation}
Alternatively, with $p$ in terms of $q$,
\begin{equation}
\label{eq:p_of_q_inequality}
    p \geq
    \begin{cases}
        0 & \text{if $q \leq \epsilon(\tau)^{2}$,} \\
        q + \epsilon^{2} (1-2q) - 2 \sqrt{\epsilon^{2} (1 - \epsilon^{2}) q(1-q)} & \text{if $q > \epsilon(\tau)^{2}$.}
    \end{cases}
\end{equation}

We now conjecture a number of possibilities for the error function $\epsilon$ under our proposed schedule \cref{eq:optimal_schedule} and examine the resulting relationships between $p$ and $\tau$.
Note that an error function that takes values at least that of another error function is weaker in the sense that if the adiabatic theorem holds for the latter function, then it also holds for the former.
As we are mainly interested in the behavior of the system near the beginning of its evolution, it will suffice to take examples of $\epsilon$ that are increasing functions.
Hence, we can replace $\epsilon(\tau)$ with $\epsilon(q)$ in \cref{eq:adiabatic_error,eq:overlap_in_pq,eq:q_of_p_inequality,eq:p_of_q_inequality} since the fact that $\tau \leq q$ of \cref{eq:tau_of_q} implies $\epsilon(\tau) \leq \epsilon(q)$.

\subsubsection{Constant error}
\label{sec:const_err}

This is the weakest possible choice of error function and is defined by
\begin{equation}
\label{eq:error_func_const}
    \epsilon(\tau) = \varepsilon .
\end{equation}
The adiabatic theorem (\cref{thm:adiabatic_theorem}) is expected to hold with this function in most cases (aside from those with pathological choices of Hamiltonians).
Solving $q \leq \epsilon(\tau)^{2} = \varepsilon^{2}$ using \cref{eq:q_of_tau_proposed} gives $0 \leq \tau \leq \varepsilon^{2} - \sqrt{\frac{\varepsilon^{2} (1 - \varepsilon^{2})}{N-1}}$ if $N \geq 1 / \varepsilon^{2}$ and no solutions for $\tau$ otherwise.
Thus, for all $\tau \in [0,1]$,
\begin{equation}
\label{eq:p_of_tau_exact_const_err}
    p \geq
    \begin{cases}
        0 & \text{if $\tau \leq \varepsilon^{2} - \sqrt{\frac{\varepsilon^{2} (1 - \varepsilon^{2})}{N-1}}$,} \\
        q + \varepsilon^{2} (1-2q) - 2 \sqrt{\varepsilon^{2} (1 - \varepsilon^{2}) q(1-q)} & \text{otherwise,}
    \end{cases}
\end{equation}
where in the second case $q$ stands for that given by \cref{eq:q_of_tau_proposed} in terms of $\tau$.
Example plots for \cref{eq:p_of_tau_exact_const_err} are shown in \cref{fig:p_tau_const}.
We can loosely bound this by a linear function independent of $N$ (see \cref{sec:linear_bound}) by
\begin{equation}
\label{eq:p_of_tau_loose_const_err}
    p \geq
    \begin{cases}
        0 & \text{if $\tau \leq \varepsilon$,} \\
        \tau - \varepsilon & \text{otherwise.}
    \end{cases}
\end{equation}
Using \cref{eq:final_time} this implies
\begin{equation}
\label{eq:time_bound_const_err}
    t \leq 2 \sqrt{N-1} \mathopen{}\left( 1 + \frac{p}{\varepsilon} \right)\mathclose{} .
\end{equation}

\begin{figure}
    \centering





    \includegraphics{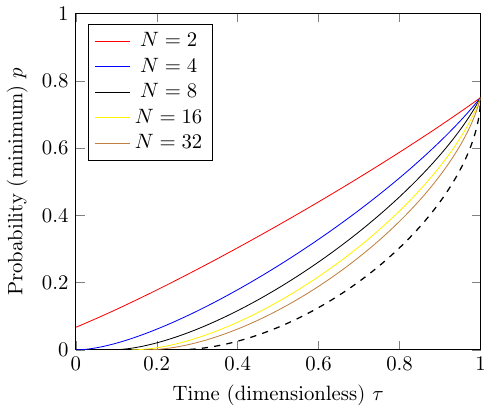}
    \caption{
        Constant error function \cref{eq:error_func_const}.
        Lower bound on the physical probability $p$ of measuring the marked state if adiabatic evolution is stopped at dimensionless time $\tau$.
        Plotted is \cref{eq:p_of_tau_exact_const_err} for $\varepsilon = 0.5$ and different search domain sizes $N$.
        The dashed line shows the asymptotic behavior as $N \to \infty$.
        In this limit of large $N$, the probability $p$ is exactly $0$ for all $\tau \leq \varepsilon^{2}$.
        In fact, it is not possible to bound $p$ below by a nonzero function proportional to $\tau$ unless $N < 1 / \varepsilon^{2}$, which is the case only for $N=2$ in the examples shown (with $\varepsilon = 0.5$).
    }
    \label{fig:p_tau_const}
\end{figure}

\subsubsection{Square-root error}
\label{sec:sqrt_err}

We define this error function by
\begin{equation}
\label{eq:error_func_sqrt}
    \epsilon(\tau) =
    \begin{cases}
        \sqrt{\tau} & \text{if $0 \leq \tau \leq \varepsilon^{2}$,} \\
        \varepsilon & \text{otherwise.}
    \end{cases}
\end{equation}
Since the second case (for $\tau > \varepsilon^{2}$) is simply that of constant error already examined in \cref{sec:const_err}, let us first focus on the first case, where $\epsilon(\tau) = \sqrt{\tau}$ for $\tau \leq \varepsilon^{2}$.
Due to \cref{eq:q_of_tau_proposed}, we know that $q \leq \tau = \epsilon(\tau)^{2}$ has no solution on $\tau \leq \varepsilon^{2} < 1$.
On the other hand, for $\tau > \varepsilon^{2}$ we are in the constant error function regime of \cref{sec:const_err}, and so overall \cref{eq:p_of_q_inequality} becomes
\begin{equation}
\label{eq:p_of_tau_exact_sqrt_err}
    p \geq
    \begin{cases}
        q + \tau (1-2q) - 2 \sqrt{\tau (1 - \tau) q(1-q)} & \text{if $\tau \leq \varepsilon^{2}$,} \\
        q + \varepsilon^{2} (1-2q) - 2 \sqrt{\varepsilon^{2} (1 - \varepsilon^{2}) q(1-q)} & \text{otherwise.}
    \end{cases}
\end{equation}
\Cref{eq:p_of_tau_exact_sqrt_err} is shown in \cref{fig:p_tau_sqrt}.
As $N \to \infty$, this bound becomes that of the constant error function \cref{eq:p_of_tau_exact_const_err} already discussed.

\begin{figure}
    \centering
        





    \includegraphics{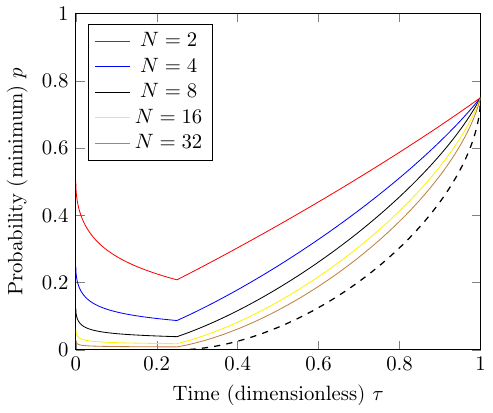}
    \caption{
        Square-root error function \cref{eq:error_func_sqrt}.
        Lower bound on the physical probability $p$ of measuring the marked state if adiabatic evolution is stopped at dimensionless time $\tau$.
        Plotted is \cref{eq:p_of_tau_exact_sqrt_err} for $\varepsilon = 0.5$ and different search domain sizes $N$.
        The dashed line shows the asymptotic behavior as $N \to \infty$.
        In this limit of large $N$, at $\tau \to 0$ the bound converges to that of the constant error function \cref{eq:p_of_tau_exact_const_err}.
    }
    \label{fig:p_tau_sqrt}
\end{figure}

\subsubsection{Sine--square-root error}
\label{sec:sinsqrt_err}

We define this error function by
\begin{equation}
\label{eq:error_func_sinsqrt}
    \epsilon(\tau) =
    \begin{cases}
        \varepsilon \sin \mathopen{}\left( \frac{\sqrt{\tau}}{\varepsilon} \right)\mathclose{} & \text{if $0 \leq \tau \leq \frac{\pi^{2}}{4} \varepsilon^{2}$,} \\
        \varepsilon & \text{otherwise.}
    \end{cases}
\end{equation}
For this function, the condition for the first case of \cref{eq:p_of_q_inequality} $q \leq \epsilon(\tau)^{2}$ has no solutions for $\tau$ (unless $N \to \infty$ in which case $\tau = 0$ is the only solution).
Thus, we only need to consider the second case, and so
\begin{equation}
\label{eq:p_of_tau_exact_sinsqrt_err}
    p \geq
    \begin{cases}
        q + \varepsilon^{2} \sin^{2} \mathopen{}\left( \frac{\sqrt{\tau}}{\varepsilon} \right)\mathclose{} (1-2q) - 2 \sqrt{\varepsilon^{2} \sin^{2} \mathopen{}\left( \frac{\sqrt{\tau}}{\varepsilon} \right)\mathclose{} \mathopen{}\left[ 1 - \varepsilon^{2} \sin^{2} \mathopen{}\left( \frac{\sqrt{\tau}}{\varepsilon} \right)\mathclose{} \right]\mathclose{} q(1-q)} & \text{if $\tau \leq \frac{\pi^{2}}{4} \varepsilon^{2}$,} \\
        q + \varepsilon^{2} (1-2q) - 2 \sqrt{\varepsilon^{2} (1 - \varepsilon^{2}) q(1-q)} & \text{otherwise,}
    \end{cases}
\end{equation}
where again $q$ stands for that given by \cref{eq:q_of_tau_proposed} in terms of $\tau$.
Examples of \cref{eq:p_of_tau_exact_sinsqrt_err} are shown in \cref{fig:p_tau_sinsqrt}.
Note that for very small $\tau \ll \varepsilon^{2}$, this error function \cref{eq:error_func_sinsqrt} is $\epsilon(\tau) \approx \sqrt{\tau}$, i.e. the square-root error function \cref{eq:error_func_sqrt} already discussed.

Under this error function \cref{eq:p_of_tau_exact_sinsqrt_err} implies that $p \geq \frac{1}{4N}$ for all $\tau \in [0,1]$.
We can form a loose bound by combining it with the linear bound \cref{eq:p_of_tau_loose_const_err} derived for the generic constant error function.
This yields a linear bound that is nowhere zero
\begin{equation}
    p \geq
    \begin{cases}
        \frac{1}{4N} & \text{if $\tau \leq \varepsilon + \frac{1}{4N}$,} \\
        \tau - \varepsilon & \text{otherwise.}
    \end{cases}
\end{equation}

\begin{figure}
    \centering





    \includegraphics{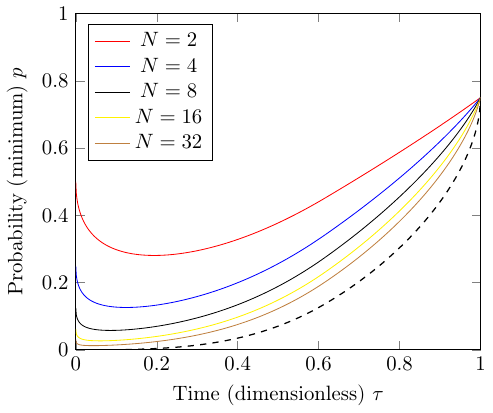}
    \caption{
        Sine--square-root error function \cref{eq:error_func_sinsqrt}.
        Lower bound on the physical probability $p$ of measuring the marked state if adiabatic evolution is stopped at dimensionless time $\tau$.
        Plotted is \cref{eq:p_of_tau_exact_sinsqrt_err} for $\varepsilon = 0.5$ and different search domain sizes $N$.
        The dashed line shows the asymptotic behavior as $N \to \infty$.
        In this limit of large $N$, at $\tau \to 0$ the probability $p$ vanishes with vanishing gradient.
        It is not possible to bound $p$ below by a nonzero function proportional to $\tau$ as $N \to \infty$.
    }
    \label{fig:p_tau_sinsqrt}
\end{figure}

\subsubsection{Scaled square-root error}
\label{sec:sqrt_scaled_err}

This is the strongest potential error function we consider, and is defined by
\begin{equation}
\label{eq:error_func_sqrt_scaled}
    \epsilon(\tau) = \varepsilon \sqrt{\tau} .
\end{equation}
The condition $q \leq \epsilon(\tau)^{2}$ in \cref{eq:p_of_q_inequality} has no solutions for $\tau$ (unless $N \to \infty$ in which case $\tau = 0$ is the only solution).
Inserting \cref{eq:error_func_sqrt_scaled} into \cref{eq:p_of_q_inequality} thus gives
\begin{equation}
\label{eq:p_of_tau_exact_sqrt_scaled_err}
    p \geq q + \varepsilon^{2} \tau (1-2q) - 2 \varepsilon \sqrt{\tau (1 - \varepsilon^{2} \tau) q(1-q)} ,
\end{equation}
where again $q$ stands for that given by \cref{eq:q_of_tau_proposed} in terms of $\tau$.
Example plots for \cref{eq:p_of_tau_exact_sqrt_scaled_err} are shown in \cref{fig:p_tau_sqrt_scaled}.

\begin{figure}
    \centering






    \includegraphics{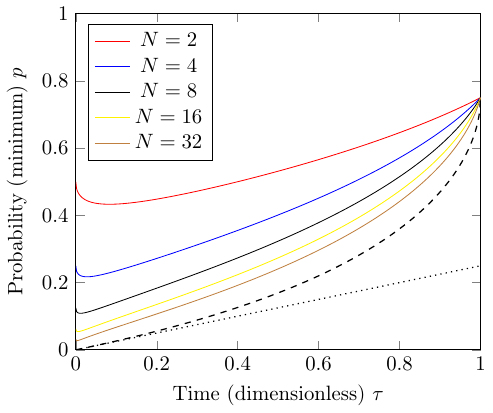}
    \caption{
        Scaled square-root error function \cref{eq:error_func_sqrt_scaled}.
        Lower bound on the physical probability $p$ of measuring the marked state if adiabatic evolution is stopped at dimensionless time $\tau$.
        Plotted is \cref{eq:p_of_tau_exact_sqrt_scaled_err} for $\varepsilon = 0.5$ and different search domain sizes $N$.
        The dashed line shows the asymptotic behavior as $N \to \infty$.
        For all $N$, the probability is bounded below by $p \geq (1 - \varepsilon)^{2} \tau$ as shown by the dotted line.
    }
    \label{fig:p_tau_sqrt_scaled}
\end{figure}

As noted earlier, since $\tau \leq q$ thanks to \cref{eq:tau_of_q} and that $\epsilon$ is increasing, we can replace $\tau$ with $q$ for a weaker bound and so
\begin{equation}
\begin{split}
    p
    & \geq q + \varepsilon^{2} q (1 - 2q) - 2 \varepsilon q \sqrt{(1 - \varepsilon^{2} q)(1-q)} \\
    & = (1 - \varepsilon)^{2} q + 2 \varepsilon q \mathopen{}\left[ 1 - \varepsilon q - \sqrt{(1 - \varepsilon^{2} q)(1-q)} \right]\mathclose{} .
\end{split}
\end{equation}
Notice that
\begin{equation}
\begin{split}
    \sqrt{(1 - \varepsilon^{2} q)(1-q)}
    & = \sqrt{1 + \varepsilon^{2} q^{2} - (1 + \varepsilon^{2})q} \\
    & \leq \sqrt{1 + \varepsilon^{2} q^{2} - 2 \varepsilon q} \\
    & = 1 - \varepsilon q .
\end{split}
\end{equation}
and therefore that we have
\begin{equation}
\label{eq:p_prop_q_sqrt_scaled}
    p \geq (1 - \varepsilon)^{2} q .
\end{equation}
Again since $\tau \leq q$, this becomes
\begin{equation}
    p \geq (1 - \varepsilon)^{2} \tau .
\end{equation}
Assuming that $p \leq (1 - \varepsilon)^{2}$, we can also rearrange this for $\tau$ giving
\begin{equation}
    \tau \leq \frac{p}{(1 - \varepsilon)^{2}} .
\end{equation}
By \cref{eq:final_time}, to achieve marked probability at least $p$, a sufficient physical time for evolution would be
\begin{equation}
\label{eq:time_bound_sqrt_scaled_err}
    t = T \tau \leq \frac{2 \sqrt{N-1}}{\varepsilon (1 - \varepsilon)^{2}} p
\end{equation}
under the square-root error function.

\section{Exact adiabaticity and probability}
\label{sec:exact_adiabaticity}

By solving the Schr{\"o}dinger equation of \cref{eq:schrodinger_dimensionless}, we can to find the exact physical probability of measuring the marked state and the error,
\begin{align}
    p(\tau) & = \lvert \braket{m}{\psi(\tau)} \rvert^{2} , \\
    \epsilon(\tau) & = \sqrt{1 - \lvert \braket{\varepsilon_{0}(\tau)}{\psi(\tau)} \rvert^{2}} ,
\end{align}
respectively.
For both quantities simultaneously, the Schr{\"o}dinger equation can be reduced to a two-dimensional homogeneous linear system of ordinary differential equations (see \cref{sec:numerical_schrodinger} for details).
The system to be solved is
\begin{equation}
\label{eq:ode_equations}
    \chi^{\prime}(\tau) = -i T
    \begin{bmatrix}
         1 - s(\tau) & - \frac{1-s(\tau)}{N} \\
        -s(\tau) & s(\tau)
    \end{bmatrix}
    \chi(\tau)
\end{equation}
with initial conditions
\begin{equation}
\label{eq:ode_initialconds}
    \chi(0) =
    \begin{bmatrix}
        N^{- \frac{1}{2}} \\
        N^{\frac{1}{2}}
    \end{bmatrix} .
\end{equation}
Here, $\chi(\tau)$ is a two-dimensional column vector and $s$ is the schedule to be used (for example, our proposed schedule given by \cref{eq:optimal_schedule}).
The desired quantities at each time $\tau$ are then calculated from the components of the solution $\chi(\tau)$.
We approach the task of solving this initial value problem numerically, using an explicit Runge--Kutta method of order $5(4)$ \cite{kutta1901beitrag,dormand1980family,shampine1986some}.
Specifically, we make use of the method implemented in the SciPy Python library with greater than default precision (\texttt{atol} and \texttt{rtol} both set to \num{1e-12}) \cite{virtanen2020scipy}.
This method chooses an adaptive step size (that may differ at different points of the domain) based on the specified error tolerances.
\Cref{fig:error_numerical} shows the exact (numerically obtained) error at all times.
Similarly, \cref{fig:p_numerical} shows the exact probability of measuring the marked state using our schedule.

\begin{figure}
    \centering
    \begin{subfigure}{0.5\linewidth}
        \centering
        \includegraphics[height=5.5cm]{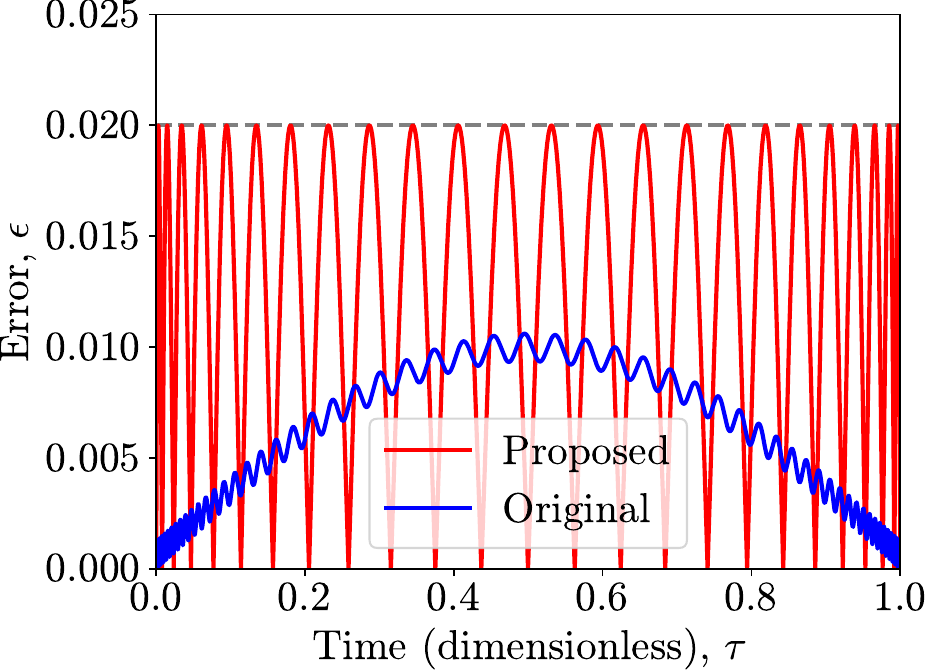}
        \caption{}
        \label{fig:error_numerical_full}
    \end{subfigure}\hfill
    \begin{subfigure}{0.5\linewidth}
        \centering
        \includegraphics[height=5.39cm]{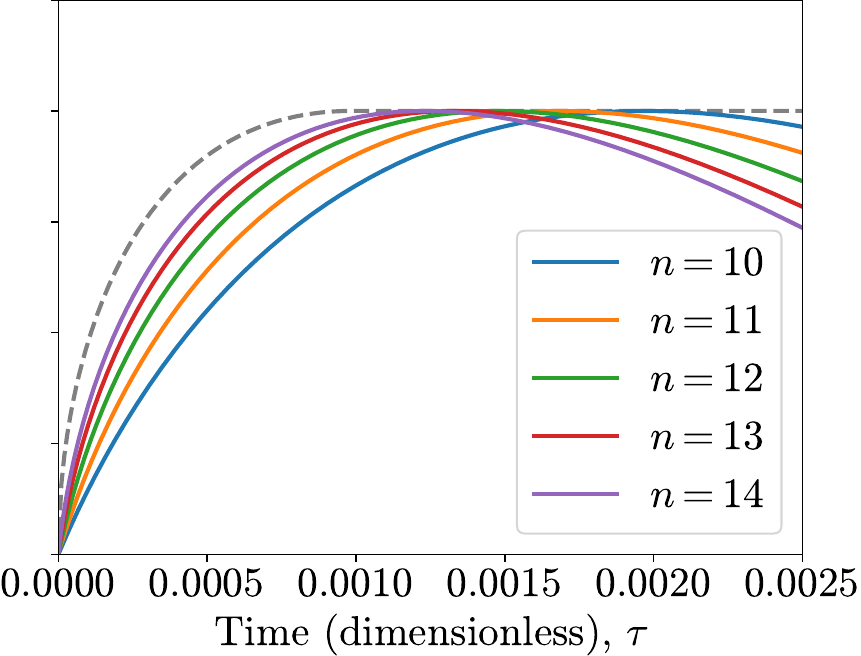}
        \caption{}
        \label{fig:error_numerical_zoomed}
    \end{subfigure}
    \caption{
        \textbf{Left~(a):}
        Exact (numerical) error $\epsilon$ at all dimensionless times $\tau$ using both the original schedule of \textcite{roland2002quantum} (blue) and our proposed schedule of \cref{eq:optimal_schedule} (red).
        For this example, parameters were set to $n = 8$ (so that $N = 256$) and $\varepsilon = 0.02$.
        That the error is bounded by the horizontal gray dashed line at $0.02$ confirms adiabaticity for both schedules as predicted by the adiabatic theorem of \cref{thm:adiabatic_theorem} with error function $\epsilon(\tau) = \varepsilon$.
        \textbf{Right~(b):}
        The same plot but only for the proposed schedule \cref{eq:optimal_schedule} with $10 \leq n \leq 14$, and showing only small $\tau$.
        Here, the gray dashed line is the conjectured sine--square-root error function \cref{eq:error_func_sinsqrt}, to which the error appears to converge to from below as $N \to \infty$.
    }
    \label{fig:error_numerical}
\end{figure}

We observe from \cref{fig:error_numerical_full} the overall form of the exact error and numerical evidence that the adiabatic theorem \cref{thm:adiabatic_theorem} holds for our proposed schedule with an error function $\epsilon(\tau) = \varepsilon$.
We see that the exact error quantity oscillates with constant peak amplitudes and with increasing frequency away from $\tau = 1/2$.
In \cref{fig:error_numerical_zoomed} we increase the scale to show only the first peaks of the error.
We see that the error for this initial peak seems to be asymptotically bounded above for all search domain sizes $N$ by the sine--square-root error function \cref{eq:error_func_sinsqrt} that we considered in \cref{sec:sinsqrt_err}.
However, the error is not bounded by the scaled square-root error function of \cref{sec:sqrt_scaled_err}, which would give rise to the physical time required for evolution scaling proportional to desired probability $p$.
\Cref{thm:ideal_parallelization} (allowing for complete efficient parallelization of unstructured search using our schedule in the ideal, errorless adiabatic regime) unfortunately, then, does not appear to be physically realizable.
From here on, we shall therefore assume the sine--square-root error function of \cref{sec:sinsqrt_err} (or the weaker error functions of \cref{sec:const_err,sec:sqrt_err}) in non-numerical analysis.

\begin{figure}
    \centering
    \begin{subfigure}{0.5\linewidth}
        \centering
        \includegraphics[height=5.5cm]{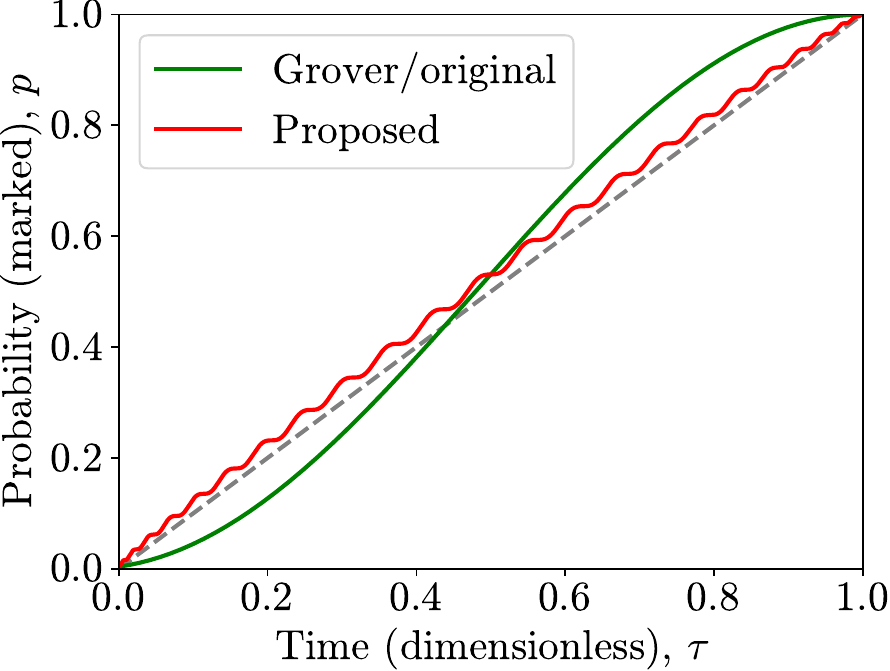}
        \caption{}
        \label{fig:p_numerical_full}
    \end{subfigure}\hfill
    \begin{subfigure}{0.5\linewidth}
        \centering
        \includegraphics[height=5.5cm]{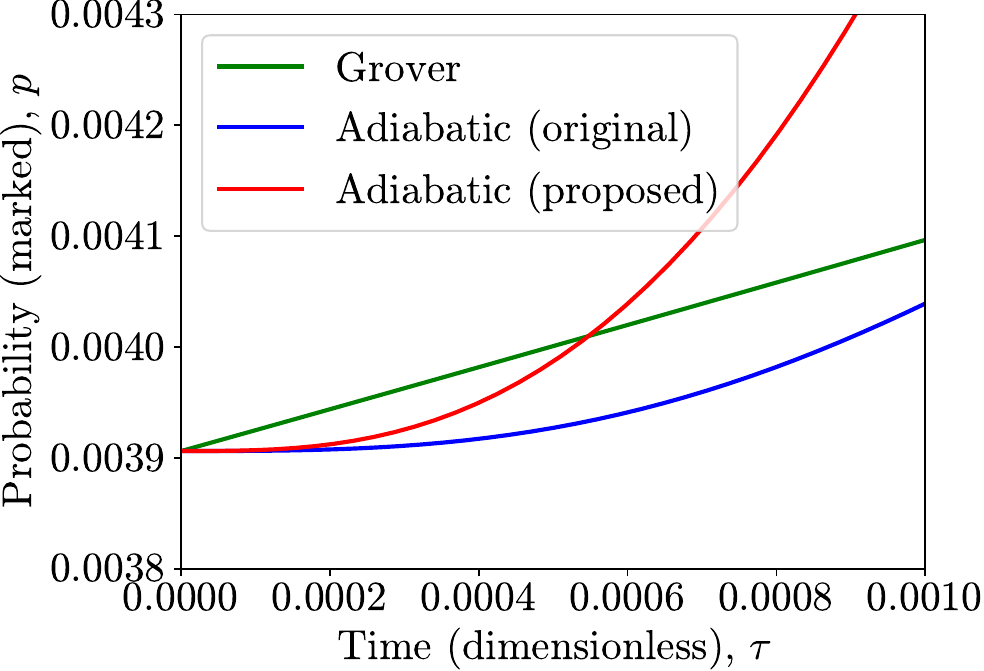}
        \caption{}
        \label{fig:p_numerical_zoomed}
    \end{subfigure}
    \caption{
        \textbf{Left~(a):}
        Exact (numerical) probability for measuring the marked state at all dimensionless times $\tau$.
        Shown are the plots using the proposed schedule \cref{eq:optimal_schedule} (red), the original schedule of \textcite{roland2002quantum}, and a comparison to Grover's algorithm.
        For this example, parameters were set to $N = 256$ and $\varepsilon = 0.02$.
        The curve of the original schedule exhibits some oscillations about that for Grover's algorithm, but these oscillations are too small to be seen in the plot, and thus both have been colored green.
        The gray dashed line is the identity function, to which the proposed (red) curve oscillates about for $N \to \infty$ and further converges to as both $n \to \infty$ and $\varepsilon \to 0$.
        \textbf{Right~(b):}
        The same plot for very small $\tau$.
    }
    \label{fig:p_numerical}
\end{figure}

From \cref{eq:time_bound_const_err} and the fact that we need not evolve the system in time at all to measure the marked state with probability $1/N$ (the initial state is prepared in uniform superposition), the physical time needed to elapse to attain probability $p$ is
\begin{equation}
\label{eq:protocol_prob}
    t \leq
    \begin{cases}
        0 & \text{if $p \leq \frac{1}{N}$,} \\
        2 \sqrt{N-1} \mathopen{}\left( 1 + \frac{p}{\varepsilon} \right)\mathclose{} & \text{otherwise.}
    \end{cases}
\end{equation}
For simplicity, from this rather generic bound, we form \cref{prot:prob_p_measurement}.
The protocol concretely describes a method to successfully identify the marked state with probability $p$ in at most this much time (along with some constant overhead time).

\Cref{fig:p_numerical} shows how the marked state probability evolves over time.
In \cref{fig:p_numerical_full} we see an oscillating probability of measuring the marked state that increases almost linearly as expected (compare with \cref{eq:q_of_tau_proposed}, which is the errorless, and therefore oscillation-free $\varepsilon \to 0$ case).
The probability remains greater than using the original schedule for all times until approximately halfway through the adiabatic evolution.
Note that we also include a comparison to probability using Grover's algorithm in \cref{fig:p_numerical} (and that the original adiabatic schedule follows Grover's algorithm almost exactly).
Let us now discuss this comparison further.

\subsection{Comparison with Grover's algorithm}
\label{sec:exact_grover_comparison}

In both the adiabatic case and that of Grover's algorithm, $\tau$ is a dimensionless quantity representing the fractional completion of the ``evolution'' parts of the algorithms (Hamiltonian interpolation in the adiabatic case and the application of successive iterations in Grover's algorithm).
The marked state probability using Grover's algorithm as a function of $\tau$ is given by (see \cref{sec:grover_probability})
\begin{equation}
\label{eq:grover_prob_tau}
    p_{\text{g}} = \cos^{2} \mathopen{}\left[ (1 - \tau) \arccos{\frac{1}{\sqrt{N}}} \right]\mathclose{} .
\end{equation}
This is a fair comparison if we assume that all other parts of the algorithms take negligible time, and that each Grover iteration takes an identical fixed amount of physical time, independent of $N$.
We can ensure a fair comparison of their physical (truncated) running times by setting their durations equal
\begin{equation}
\label{durations_equal}
    T_{\text{g}} = k T_{\text{a}} ,
\end{equation}
where the subscripts refer to Grover's and adiabatic algorithms, respectively, and $k$ is a constant denoting the number of Grover iterations that are performed per unit of time in the Schr{\"o}dinger equation of \cref{eq:schrodinger}.
Comparing between adiabatic evolution (using our proposed schedule) and Grover's algorithm this means setting
\begin{equation}
\label{eq:grover_diabaticity}
    \varepsilon
    = \frac{8k \sqrt{N_{\text{a}} - 1}}{\frac{\pi}{\arccsc{\sqrt{N_{\text{g}}}}} - 2}
    \approx \frac{8k}{\pi} \sqrt{\frac{N_{\text{a}}}{N_{\text{g}}}}
\end{equation}
in our adiabatic algorithm or, letting $N = N_{\text{a}} = N_{\text{g}}$, it means setting
\begin{equation}
\label{eq:grover_diabaticity_equaldomain}
    \varepsilon = \frac{8k \sqrt{N - 1}}{\frac{\pi}{\arccsc{\sqrt{N}}} - 2} \approx \frac{8k}{\pi} ,
\end{equation}
where the approximations are for large search domain sizes $N_{\text{a}}$ and $N_{\text{g}}$.

As noted for the original schedule, the probability using our adiabatic schedule appears to remain greater than for Grover's algorithm for approximately the first half of evolution.
\Cref{fig:p_numerical_zoomed} exhibits the exception to this, which is at very small dimensionless times, when the strength of our initial oscillations can play a role.
In this region, we observe that for a short initial time the adiabatic algorithms yield a smaller probability before increasing above Grover's algorithm.
Let us denote by $\tau_{\text{cross}}$ the dimensionless time at which our adiabatic schedule crosses into the region in which it outperforms Grover's algorithm in terms of probability (until approximately halfway through the evolution).
That is, $p_{\text{a}} \geq p_{\text{g}}$ whenever $\tau_{\text{cross}} \leq \tau \lesssim 1/2$.
Translating into physical time, letting
\begin{equation}
    t_{\text{cross}} = T \tau_{\text{cross}} = \frac{2}{\varepsilon} \sqrt{N-1} \tau_{\text{cross}}
\end{equation}
according to \cref{eq:final_time}, this occurs whenever
\begin{equation}
    t_{\text{cross}} \leq t \lesssim \frac{\sqrt{N-1}}{\varepsilon} .
\end{equation}
To be clear, here physical time is written with time units as in \cref{eq:schrodinger}, and we have assumed that maximum durations (in these units) of Grover's algorithm and our schedule are equal as in \cref{durations_equal}.

By utilizing the same numerical solving methods used to find the probability $p_{\text{a}}$ of our adiabatic schedule as a function of $\tau$ in \cref{fig:p_numerical} and comparing to the probability $p_{\text{g}}$ of Grover's algorithm stated in \cref{eq:grover_prob_tau}, we can numerically evaluate the dimensionless crossing time $\tau_{\text{cross}}$ for any fixed diabaticity $\varepsilon$ and hence find $t_{\text{cross}}$.
Specifically, we find the largest root of the difference $p_{\text{g}}(\tau) - p_{\text{a}}(\tau)$ where $\tau \ll \frac{1}{2}$.
The results of this are shown in \cref{fig:crossing_time}.
It appears that for large $N \to \infty$ and a constant diabaticity parameter $\varepsilon$ (i.e. fixed by \cref{eq:grover_diabaticity_equaldomain} based on the rate of Grover iterations achievable $k$) the physical crossing time is decreasing in $N$ to some constant
\begin{equation}
    t_{\text{cross}} = \bigO(1).
\end{equation}
In other words, irrespective of the (assumed large enough) search domain size, our adiabatic schedule outperforms Grover's algorithm in terms of probability when the physical time $t$ is truncated between some constant $t_{\text{cross}}$ and approximately half of a full evolution $\frac{\pi}{8k} \sqrt{N}$.

\begin{figure}
    \centering
    \includegraphics[width=0.5\linewidth]{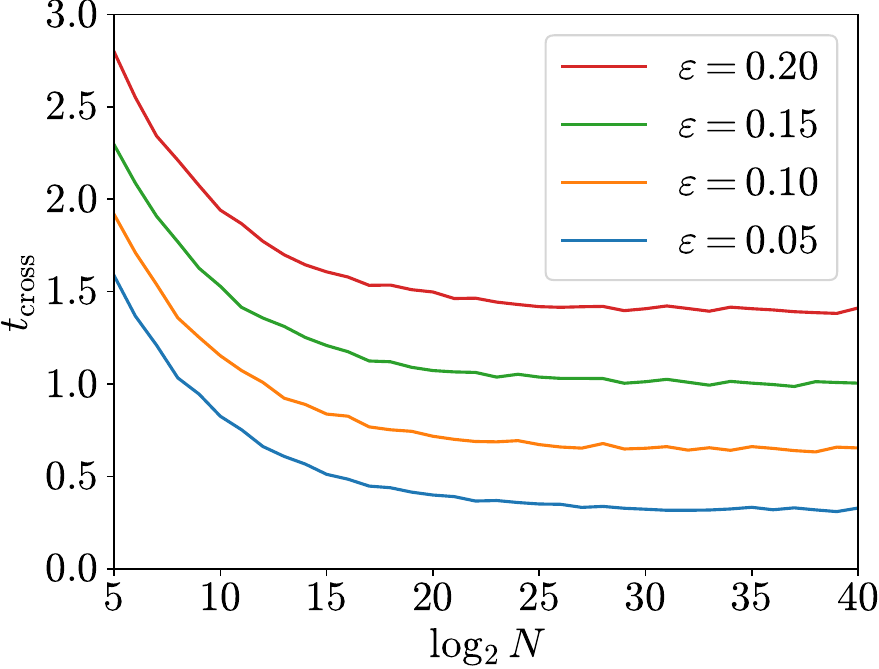}
    \caption{
        Physical time $t_{\text{cross}}$ (in the time units of \cref{eq:schrodinger}) after which the probability of measuring the marked state becomes greater for our schedule than for Grover's algorithm as a function of search domain size $N$ (plotted on a logarithmic scale).
        For fixed diabaticity $\varepsilon$ (which is set by \cref{eq:grover_diabaticity_equaldomain} for fair comparison), this appears to be decreasing and $t_{\text{cross}} = \bigO(1)$ as $N \to \infty$.
        Results were obtained by numerical root finding by comparison of \cref{eq:grover_prob_tau} with $p_{\text{a}}(\tau)$ solved using the numerical techniques discussed in \cref{sec:exact_adiabaticity}.
        That the curves are not quite smooth is likely due to relatively coarse step sizes being chosen by the ODE solver compared to the small dimensionless time scales of roots.
    }
    \label{fig:crossing_time}
\end{figure}

\section{Bounded resources}
\label{sec:bounded_resources}

As we are dealing with incomplete runs of search algorithms, it is natural to ask whether this can be leveraged in cases where the hardware resources to which we have access is in some way bounded.
For example, let us consider that we have access to $S$ independent quantum processors simultaneously at any one moment.
We may demand that any adiabatic evolution may only continue for some maximum physical time $t_{\text{c}} \leq T$ (in the case of Grover's algorithm, we would similarly constrain the maximum depth of the computation).
This may be due, for example, to the quantum hardware having a constrained coherence time.

We would like to know how many independent runs of \cref{prot:prob_p_measurement} are required if the evolution time of the algorithm is constrained to a maximum physical time $t_{\text{c}}$.
That is, how many runs we must perform in order that we fail to measure the marked state in any of the runs with probability at most some small $\alpha > 0$ of our choice.
Then, the marked state would be classically verified with probability at least $\gamma = 1 - \alpha$ upon completion of all the runs.
Running the protocol with the maximum possible choice of probability
\begin{equation}
    p = \mathopen{}\left( \frac{t_{\text{c}}}{2 \sqrt{N-1}} - 1 \right)\mathclose{} \varepsilon
\end{equation}
attains an exactly $t_{\text{c}}$ cutoff time in the protocol.
Note that we require coherence time of at least
\begin{equation}
\label{eq:advantage_coherence_time}
    t_{\text{c}} > 2\sqrt{N-1} \mathopen{}\left( 1 + \frac{1}{\varepsilon N} \right)\mathclose{}
\end{equation}
time units using \cref{prot:prob_p_measurement} to gain an advantage over randomly guessing the marked state, since otherwise the marked state output probability is $p = 1/N$ and the protocol reduces to outputting one of the $N$ possible states uniformly at random after a constant time.
This minimum time can likely be reduced to a constant with tighter consideration of all the bounds (such as by using the sine--square-root error function of \cref{sec:sinsqrt_err} or some stronger function for finite $N$), as \cref{fig:crossing_time} indicates that our schedule yields higher probabilities than Grover's algorithm after constant time depending only on $\varepsilon$.
Since the independent runs form a binomial process with success probability $p$, we expect to obtain the marked state exactly once in every $1/p$ runs (setting aside that this may not be an integer).
In $r = m/p$ runs for any $m \geq p$, the failure probability is $\alpha = (1 - p)^{m/p} \leq \euler^{-m}$.
Hence, performing $r \geq -p^{-1} \ln{\alpha}$ runs guarantees probability at most $\alpha$ of failing to measure the marked state in all of them.
As $r$ must be an integer, we will set
\begin{equation}
    r = \mathopen{}\left\lceil \frac{- \ln{\alpha}}{p} \right\rceil\mathclose{} .
\end{equation}

Given that we have access to at most $S$ independent quantum processors (which we may run in parallel) with coherence time $t_{\text{c}}$ satisfying \cref{eq:advantage_coherence_time}, the overall running time (assuming that the steps of \cref{prot:prob_p_measurement} taking constant time amount to time $c$) is
\begin{equation}
    \mathopen{}\left\lceil \frac{r}{S} \right\rceil\mathclose{} (t_{\text{c}} + c)
    = \mathopen{}\left\lceil \frac{1}{S} \mathopen{}\left\lceil \frac{-2 \varepsilon^{-1} \ln{\alpha}}{t_{\text{c}} (N-1)^{-\frac{1}{2}} - 2} \right\rceil\mathclose{} \right\rceil\mathclose{} (t_{\text{c}} + c) .
\end{equation}
Suppose, for example, we had $t_{\text{c}} = 2(b+1) \sqrt{N}$, where $0 < b \leq \frac{1}{\varepsilon} \sqrt{1 - 1/N} - 1$, so that $p \approx b \varepsilon$ for large enough $N$.
Then the running time for large $N$ becomes
\begin{equation}
\label{eq:proposed_overall_runtime}
    2(b+1) \mathopen{}\left\lceil \frac{1}{S} \mathopen{}\left\lceil \frac{- \ln{\alpha}}{b \varepsilon} \right\rceil\mathclose{} \right\rceil\mathclose{} \sqrt{N}
    = \bigO(\sqrt{N}) ,
\end{equation}
with any additional processors $S$ giving rise to a speedup by a constant factor.
Otherwise, for smaller $t_{\text{c}}$ than \cref{eq:advantage_coherence_time} such that $p = 1/N$, the overall running time is
\begin{equation}
    \mathopen{}\left\lceil \frac{r}{S} \right\rceil\mathclose{} c
    = \mathopen{}\left\lceil \frac{\lceil -N \ln{\alpha} \rceil}{S} \right\rceil\mathclose{} c
    = \bigO \mathopen{}\left( \mathopen{}\left\lceil \frac{N}{S} \right\rceil\mathclose{} \right)\mathclose{} .
\end{equation}
In this regime there is no quantum advantage, although classical parallelization is exhibited.

\subsection{Grover's algorithm with bounded depth}
\label{sec:bounded_depth_grover}

A bounded coherence time of $t_{\text{c}}$ corresponds in Grover's algorithm to a maximum of $kt_{\text{c}}$ iterations, where $k$ is the number of Grover iterations achievable per unit of time in the Schr{\"o}dinger equation of \cref{eq:schrodinger}.
The probability of the marked state at this time is then
\begin{equation}
    p = \cos^{2} \mathopen{}\left[ \mathopen{}\left( 1 - \frac{2kt_{\text{c}}}{\frac{\pi}{2 \arccsc{\sqrt{N}}} - 1} \right)\mathclose{} \arccos{\frac{1}{\sqrt{N}}} \right]\mathclose{}
    \approx \sin^{2} \mathopen{}\left( \frac{2kt_{\text{c}}}{\sqrt{N}} \right)\mathclose{} ,
\end{equation}
where the approximation is for large $N$.
If $t_{c}$ were to be constant in $N$, then for large $N$ the probability would become
\begin{equation}
    p \approx \frac{4 k^{2} t_{\text{c}}^{2}}{N} = \bigO \mathopen{}\left( \frac{1}{N} \right)\mathclose{} .
\end{equation}
Thus, the overall running time for a constant constrained coherence time using Grover's algorithm would be $\bigO(\lceil N/S \rceil)$, offering no quantum advantage.
As mentioned, we expect based on \cref{fig:p_tau_sinsqrt,fig:crossing_time} that for our proposed adiabatic schedule $t_{c}$ can be made constant while still maintaining quantum advantage as in \cref{eq:proposed_overall_runtime}.

\section{Discussion}
\label{sec:discussion}

In this work, we introduced an alternative time schedule to interpolate between initial and final Hamiltonians for the quantum unstructured search problem using an adiabatic quantum computer.
We examined some resulting properties in comparison to existing algorithms (in particular, adiabatic unstructured search using the original schedule of \textcite{roland2002quantum}, as well as Grover's algorithm \cite{grover1996fast,boyer1998tight}).
Comparatively, the schedule given in \cref{eq:optimal_schedule} and shown in \cref{fig:schedule_comparison}, allows us to move even faster at the beginning and end times of the interpolation, leaving additional time to evolve slowly at the most important middle times where the energy gap is small.
This trading of speed between the middle and end times maintains the expected quadratic speedup of unstructured search over classical algorithms.

We were particularly interested in whether this schedule would lead to interesting behavior at intermediate times, which could be exploited by canceling the evolution of the algorithm early, accepting a reduced probability of successfully identifying the target state upon measuring the system at this time.
A distinguishing property of this schedule arises first in the ``errorless'' ideal adiabatic limit, where the state of the system perfectly matches the instantaneous ground state of the Hamiltonian over all times (a regime that is physically unrealizable, as it requires infinite time to complete such an evolution).
In this limit, we showed in \cref{eq:q_of_tau_proposed,eq:tau_of_q} that the probability grows in direct proportion with the dimensionless time $\tau$ (a quantity proportional to real time which can be thought of as the fractional completion of the interpolation) for all times.
This is not the case for the original adiabatic schedule or for Grover's algorithm (which behave identically to each other as shown by \cref{thm:grover_equiv}), and can be seen in \cref{fig:q_tau}.
If such probabilities could be achieved in finite physical time, it would immediately imply \cref{thm:ideal_parallelization} that time could be perfectly traded for space in unstructured search, e.g. that $\sqrt{N}$ adiabatic quantum machines running in parallel to perform unstructured search in a constant time would be a physically realizable feat.

As the adiabatic theorem (\cref{thm:adiabatic_theorem}) only guarantees that the system remains \emph{close} to the instantaneous ground state over the finite course of a real algorithm, we examined generically in \cref{sec:approx_adiabatic} the extent to which this state deviation is reflected in the probabilities in the worst case.
Of the several progressively stronger candidate upper bounds on the distance (error functions) we conjectured, we showed (in \cref{sec:const_err}) that the strongest of these---in which the fidelity $\lvert \braket{\varepsilon_{0}(\tau)}{\psi(\tau)} \rvert^{2} \geq 1 - \varepsilon^{2} \tau$ is allowed to decrease linearly with time---would yield perfect parallelization despite being nonzero.
In \cref{sec:exact_adiabaticity}, we derived equations for the exact probabilities and errors over time from the Schr{\"o}dinger equation.
Our numerical results (\cref{fig:error_numerical}) showed that indeed the schedule exhibits adiabaticity in the sense of \cref{thm:adiabatic_theorem} for the parameters tested (something that cannot strictly be assumed from satisfying the adiabatic condition \cref{eq:adiabatic_cond} alone due to the existence of certain pathological counterexamples and inconsistencies in the adiabatic theorem \cite{marzlin2004inconsistency,tong2005quantitative,du2008experimental,wu2008adiabatic}).
Unfortunately, it appears that perfect parallelization does not translate into practice, as the numerically evaluated errors asymptotically exceeded those of our parallelizable error function as $\tau \to 0$.
The tightest error function we considered still bounding the true errors was the ``sine--square-root'' function of \cref{sec:sinsqrt_err}, which does not admit perfect parallelizability.

Nonetheless, we numerically demonstrated that using our schedule outperforms the original schedule in terms of probability during all times earlier than approximately halfway through a full evolution (\cref{fig:p_numerical}).
The same can be said when compared to Grover's algorithm, with the exception of at very small times (\cref{fig:p_numerical_zoomed}).
We showed (\cref{fig:crossing_time}) that the (physical) time at which our adiabatic schedule begins to outperform Grover's algorithm is $\bigO(1)$ as $N \to \infty$.
In other words, after some short fixed time that is constant in $N$ (and seemingly approximately proportional to the number of Grover iterations achievable per unit of time), our schedule is advantageous over Grover's algorithm (and always over using the original adiabatic schedule) in applications of unstructured search that expect to have their evolution terminated early before around half the completion time, and otherwise run to completion.
One could envisage such scenarios in cryptographic applications, such as where there are races between parties to find elements of preimages of cryptographic hash functions.

Finally, based on our findings about the adiabaticity of interpolation under our schedule, we produced \cref{prot:prob_p_measurement} which guarantees successfully obtaining the marked state with probability at least $p$ within running time (including overhead steps amounting to a constant time $c$) of $2 \sqrt{N}(1 + p / \varepsilon) + c$.
With this protocol, we derived overall running times of our algorithm for resource-limited scenarios in which at most $S$ quantum processors may be working simultaneously in parallel, each of which has a bounded coherence time (\cref{sec:bounded_resources}).

\paragraph{Future work}

In \cref{prot:prob_p_measurement}, there exists a minimum time before which a measurement would give no better than the probability of a uniformly random guess $1/N$ to successfully obtain the marked state.
As it stands, this minimum time \cref{eq:advantage_coherence_time} increases as $\bigO(\sqrt{N})$, which is undesirable in the context of running the protocol with hardware that has constrained coherence time.
It is very reasonable to assume that this scaling in $N$ can be reduced or even removed with more careful consideration of the approximations used to arrive at \cref{eq:protocol_prob}, such as by bounding \cref{eq:p_of_tau_exact_sinsqrt_err} below by $1/N$ or by using a tighter error function that depends on $N$.
This is because it is indicated in \cref{fig:crossing_time} that our schedule yields greater probabilities than Grover's algorithm (which itself immediately improves upon uniform probability when $t > 0$) beginning at a constant time.
Moreover, the dimensionless time $\tau_{\text{min}} > 0$ at which the valid (according to \cref{fig:p_numerical_zoomed}) ``sine--square-root'' error function of \cref{sec:sinsqrt_err} gives probability exceeding uniform probability $p \geq 1/N$ appears in \cref{fig:p_tau_sinsqrt} to decrease to $\tau_{\text{min}} \to 0$ as $N \to \infty$.
We would then achieve quantum advantage for constant constrained coherence times, which as shown in \cref{sec:bounded_depth_grover} is not possible with Grover's algorithm.

While the numerical solution to the error showed in \cref{fig:error_numerical_zoomed} indicates that the exact error increases to the error function of \cref{eq:error_func_sinsqrt} as $N \to \infty$, this evidence does not constitute a mathematical proof that perfect parallelization is impossible.
While this is known for Grover's algorithm \cite{zalka1999grover}, exhibiting a proof of this specifically for adiabatic quantum computers could be illuminating.
One approach could be to attempt to either analytically solve or extract exact qualitative properties (at least at small times) about the errors or probabilities from the initial value problem governing the system \cref{eq:ode_equations,eq:ode_initialconds}, e.g., from its Magnus expansion \cite{magnus1954exponential}.
This might also resolve the problem of whether or not all troughs in the exact error oscillations for our schedule fall to exactly zero.
In our simulations this is not the case: errors for successive troughs slowly rise, though by a magnitude so small as to be imperceptible in \cref{fig:error_numerical_full}.
We are unaware of whether this is simply an artifact of finite numerical precision.

Another interesting direction entirely absent from the present work is to use different interpolation paths to reach the final target Hamiltonian.
Alternative choices could be to use other initial Hamiltonians, such as a transverse field driver Hamiltonian $H_{0} = \sum_{j=1}^{n} (1+\sigma_{\text{x}}^{(j)})/2$ \cite{morley2019quantum}.
More generally, one could take any one-dimensional path through the space of all possible Hamiltonians (not necessarily parametrized at unit speed), ending at $H_{1}$ and beginning at some $H_{0}$ whose description is independent of which state is marked \cite{rezakhani2010accuracy}.
These paths evade the analysis presented here as they may have entirely different energy spectra and eigenstates.

\section*{Acknowledgments}

S.A.A. and P.W. acknowledge funding from STFC with grant number ST/W006537/1.
P.W. also acknowledges partial support from EPSRC grants EP/T001062/1, EP/X026167/1, EP/Z53318X/1, and EP/T026715/1.

\appendix

\section{Marked probability in Grover's algorithm}
\label{sec:grover_probability}

In Grover's algorithm (assuming that $N$ is large enough), after $t$ time steps (Grover iterations) have elapsed the marked state has amplitude
\begin{equation}
    \alpha = \sin[(2t+1) \theta] ,
\end{equation}
where
\begin{equation}
    \theta = \arcsin{\frac{1}{\sqrt{N}}} .
\end{equation}
Combining these and writing the probability $q = \lvert \alpha \rvert^{2}$ of obtaining the marked state upon measurement gives
\begin{equation}
\label{eq:grover_q_of_t}
    q = \sin^{2} \mathopen{}\left[ (2t+1) \arcsin{\frac{1}{\sqrt{N}}} \right]\mathclose{} .
\end{equation}
The first solution to $q=1$ (the number of time steps after which we first obtain the marked state with certainty) occurs when time $t$ takes the value
\begin{equation}
    T = \frac{\pi}{4 \arcsin{\frac{1}{\sqrt{N}}}} - \frac{1}{2} .
\end{equation}
Reparametrizing \cref{eq:grover_q_of_t} in terms of dimensionless time defined by $\tau = t / T$ and simplifying the resulting expression yields
\begin{equation}
    q = \cos^{2} \mathopen{}\left[ (1 - \tau) \arccos{\frac{1}{\sqrt{N}}} \right]\mathclose{} .
\end{equation}
Assuming $\tau \in [0,1]$, this can also be inverted on $q \in \left[ \frac{1}{N}, 1 \right]$ as
\begin{equation}
    \tau = 1 - \frac{\arccos{\sqrt{q}}}{\arccos{\frac{1}{\sqrt{N}}}} .
\end{equation}
Finally, note that as $N \to \infty$ we get
\begin{subequations}
\begin{align}
    q & \to \sin^{2} \mathopen{}\left( \frac{\pi}{2} \tau \right)\mathclose{} , \\
    \tau & \to \frac{2}{\pi} \arcsin{\sqrt{q}} .
\end{align}
\end{subequations}

\section{Adiabatic condition transition matrix element}
\label{sec:adiabatic_numerator}

The linear interpolating Hamiltonian
\begin{equation}
\begin{split}
    H(s)
    & = (1-s) H_{0} + s H_{1} \\
    & = (1-s) (I - \ketbra{\phi}) + s (I - \ketbra{m})
\end{split}
\end{equation}
has derivative
\begin{equation}
\begin{split}
    H^{\prime}(s)
    & = H_{1} - H_{0} \\
    & = \ketbra{\phi} - \ketbra{m} .
\end{split}
\end{equation}
Note that
\begin{subequations}
\label{eq:superposition_inner_products}
\begin{align}
    \braket{\phi}{m} & = \frac{1}{\sqrt{N}} , \\
    \braket{\phi}{m^{\perp}} & = \sqrt{\frac{N-1}{N}} ,
\end{align}
\end{subequations}
and also that
\begin{equation}
\label{eq:sin}
    \sqrt{1 - c(s)^{2}} = \frac{2}{g(s)} \frac{\sqrt{N-1}}{N} (1-s) .
\end{equation}
Expanding the matrix element of interest gives
\begin{equation}
\begin{split}
    \mel{\varepsilon_{1}(s)}{H^{\prime}(s)}{\varepsilon_{0}(s)}
    & = \left( 1 - \braket{\phi}{m}^{2} + \braket{\phi}{m^{\perp}}^{2} \right)\mathclose{} \alpha(s) \beta(s) + \braket{\phi}{m}\braket{\phi}{m^{\perp}} \mathopen{}\left( \alpha(s)^{2} - \beta(s)^{2} \right) \\
    & = 2 \mathopen{}\left( 1 - \frac{1}{N} \right)\mathclose{} \alpha(s) \beta(s) + \frac{\sqrt{N-1}}{N} \mathopen{}\left( \alpha(s)^{2} - \beta(s)^{2} \right) \\
    & = \left( 1 - \frac{1}{N} \right)\mathclose{} \sqrt{1 - c(s)^{2}} + \frac{\sqrt{N-1}}{N} c(s) ,
\end{split}
\end{equation}
where the second equality uses \cref{eq:superposition_inner_products} and the final equality uses \cref{eq:gs_coefficients}.
Finally, substituting \cref{eq:cos,eq:sin} gives
\begin{equation}
    \mel{\varepsilon_{1}(s)}{H^{\prime}(s)}{\varepsilon_{0}(s)}
    = \frac{\sqrt{N-1}}{N} \frac{1}{g(s)}
\end{equation}
as required.

\section{Marked probability using the original adiabatic schedule}
\label{sec:orig_sched_prob}

The original schedule of \textcite{roland2002quantum} reads (see \cref{tab:schedule_comparison})
\begin{equation}
    s(\tau) = \frac{1}{2} \mathopen{}\left[ 1 + \frac{\tan[(2 \tau - 1) \arctan{\sqrt{N-1}}]}{\sqrt{N-1}} \right]\mathclose{} .
\end{equation}
Inserting this into \cref{eq:marked_prob} for the probability of measuring the marked state in an ideal adiabatic evolution yields
\begin{equation}
    q = \frac{1}{2} \mathopen{}\left[ 1 + \frac{1}{\sqrt{N}} \cos\mathopen{}\left( (2\tau - 1) \arctan{\sqrt{N-1}} \right)\mathclose{} + \sqrt{\frac{N-1}{N}} \sin\mathopen{}\left( (2\tau - 1) \arctan{\sqrt{N-1}} \right)\mathclose{} \right]\mathclose{} ,
\end{equation}
where we have noted $\left\lvert (2\tau - 1) \arctan{\sqrt{N-1}} \right\rvert \leq \pi / 2$ to simplify the expression.
Now let
\begin{equation}
    \theta = \arcsin{\frac{1}{\sqrt{N}}} ,
\end{equation}
which is the angle found in Grover's algorithm, and let
\begin{equation}
    \vartheta = \frac{\pi}{2} - \theta = \arccos{\frac{1}{\sqrt{N}}} = \arctan{\sqrt{N-1}} .
\end{equation}
From this angle,
\begin{equation}
    \cos{\vartheta} = \frac{1}{\sqrt{N}} ,\quad
    \sin{\vartheta} = \sqrt{\frac{N-1}{N}} .
\end{equation}
Thus, the probability
\begin{equation}
\begin{split}
    q & = \frac{1}{2} [1 + \cos{\vartheta} \cos((2\tau - 1) \vartheta) + \sin{\vartheta} \sin((2\tau - 1) \vartheta)] \\
    & = \frac{1}{2} [1 + \cos(2 (1 - \tau) \vartheta)] \\
    & = \cos^{2}[(1 - \tau) \vartheta] \\
    & = \cos^{2} \mathopen{}\left[ (1 - \tau) \arccos{\frac{1}{\sqrt{N}}} \right]\mathclose{} ,
\end{split}
\end{equation}
where the second equality used the angle difference identity $\cos(y - x) = \cos{x}\cos{y} + \sin{x}\sin{y}$ and the third equality used the double-angle formula $\cos(2x) = 2\cos^{2}{x} - 1$.

\section{Linear marked probability bound}
\label{sec:linear_bound}

Bounding further the second case of \cref{eq:p_of_tau_exact_const_err} and noting $p \geq 0$ gives for any choice of $1 < r < \frac{1}{\varepsilon^{2}}$ that
\begin{equation}
    p \geq
    \begin{cases}
        0 & \text{if $q \leq r \varepsilon^{2}$,} \\
        \frac{r-1}{r(1-r\varepsilon^{2})} (q - r\varepsilon^{2}) & \text{otherwise.}
    \end{cases}
\end{equation}
Choosing $r = \frac{1}{\varepsilon}$ results in unit gradient
\begin{equation}
    p \geq
    \begin{cases}
        0 & \text{if $q \leq \varepsilon$,} \\
        q - \varepsilon & \text{otherwise.}
    \end{cases}
\end{equation}
Since according to \cref{eq:tau_of_q} $\tau \leq q$ for all $N$, we get
\begin{equation}
    p \geq
    \begin{cases}
        0 & \text{if $\tau \leq \varepsilon$,} \\
        \tau - \varepsilon & \text{otherwise.}
    \end{cases}
\end{equation}

\section{Ground state amplitude from the Schr{\"o}dinger equation}
\label{sec:numerical_schrodinger}

Consider the Schr{\"o}dinger equation
\begin{equation}
\label{eq:schrodinger_groundamp}
    \frac{i}{T} \dv{\tau} \ket{\psi(\tau)} = H(s(\tau)) \ket{\psi(\tau)}
\end{equation}
with some schedule $s$ given, recalling also that $\ket{\psi(0)} = \ket{\phi}$ (the uniform superposition state) and
\begin{equation}
    H(s) = I - (1-s) \ketbra{\phi} - s \ketbra{m}
\end{equation}
for some $\ket{m}$ in the $n$-qubit computational basis.
In matrix form, let us write the state vector coefficients $\psi_{k}(\tau) = \braket{k}{\psi(\tau)}$, where we index with $k \in \{0, \dots, N-1\}$.
Again indexing starting at $0$, the Hamiltonian $H(s)$ has matrix coefficients
\begin{equation}
    H_{jk}(s) =
    \begin{cases}
        - \frac{1-s}{N} + 1 - s & \text{if $j = k = m$,} \\
        - \frac{1-s}{N} + 1 & \text{if $j = k \neq m$,} \\
        - \frac{1-s}{N} & \text{if $j \neq k$.}
    \end{cases}
\end{equation}
Defining $\bar{\psi}(\tau) = \sum_{k} \psi_{k}(\tau)$, the Schr{\"o}dinger equation \cref{eq:schrodinger_groundamp} then gives rise to a homogeneous linear system of $N$ ordinary differential equations
\begin{equation}
\label{eq:ode_system_full}
    \frac{i}{T} \dv{\psi_{j}}{\tau}{(\tau)} =
    \begin{cases}
        - \frac{1-s(\tau)}{N} \bar{\psi}(\tau) + \psi_{m}(\tau) - s(\tau) \psi_{m}(\tau) & \text{if $j=m$,} \\
        - \frac{1-s(\tau)}{N} \bar{\psi}(\tau) + \psi_{j}(\tau) & \text{otherwise.}
    \end{cases}
\end{equation}
with initial conditions $\psi_{j}(0) = \frac{1}{\sqrt{N}}$ and $\bar{\psi}(0) = \sqrt{N}$.
Summing over all of these equations gives
\begin{equation}
\label{eq:ode_system_sum}
    \frac{i}{T} \dv{\bar{\psi}}{\tau}{(\tau)} = s(\tau) \mathopen{}\left[ \bar{\psi}(\tau) - \psi_{m}(\tau) \right]\mathclose{} .
\end{equation}
Let us define
\begin{subequations}
\begin{gather}
    \chi(\tau) =
    \begin{bmatrix}
        \psi_{m}(\tau) \\
        \bar{\psi}(\tau)
    \end{bmatrix} , \\
    \tilde{H}(\tau) =
    \begin{bmatrix}
         1 - s(\tau) & - \frac{1-s(\tau)}{N} \\
        -s(\tau) & s(\tau)
    \end{bmatrix} .
\end{gather}
\end{subequations}
The linear system
\begin{equation}
\label{eq:ode_system_2d}
    \chi^{\prime}(\tau) = -iT \tilde{H}(\tau) \chi(\tau)
\end{equation}
of two ordinary differential equations with initial conditions
\begin{equation}
    \chi(0) =
    \begin{bmatrix}
        N^{- \frac{1}{2}} \\
        N^{\frac{1}{2}}
    \end{bmatrix}
\end{equation}
is simply the $j=m$ case of \cref{eq:ode_system_full} combined with \cref{eq:ode_system_sum}.
Since $\tilde{H}(\tau_{1})$ does not commute with $\tilde{H}(\tau_{2})$ for $\tau_{1} \neq \tau_{2}$, the vector
\begin{equation}
\label{eq:ode_system_nonsol}
    e^{-iT \int_{0}^{\tau} \tilde{H}(x) \dd{x}} \chi(0)
\end{equation}
is \emph{not} equal to the solution value $\chi(\tau)$ to \cref{eq:ode_system_2d}.
We could attempt to find the solution using techniques such as the Magnus expansion (of which the exponent in \cref{eq:ode_system_nonsol} is the first term) \cite{magnus1954exponential}.
However, in the present work we instead choose to make use of numerical techniques.

Given the solution $\chi$, whose entries are $\psi_{m}$ and $\bar{\psi}$, we may then evaluate the quantities
\begin{subequations}
\begin{align}
    \braket{m}{\psi(\tau)} & = \psi_{m}(\tau) , \\
    \braket{\phi}{\psi(\tau)} & = \frac{\bar{\psi}(\tau)}{\sqrt{N}} .
\end{align}
\end{subequations}
From these, we can finally calculate the amplitude of the physical state $\ket{\psi(\tau)}$ in the instantaneous ground state $\ket{\varepsilon_{0}(s(\tau))}$ of the Hamiltonian $H(s(\tau))$.
Using \cref{eq:marked_perp,eq:gs_vector_0},
\begin{equation}
\begin{split}
    \braket{\varepsilon_{0}(s(\tau))}{\psi(\tau)}
    & = \left[ \alpha(s(\tau)) - \frac{\beta(s(\tau))}{\sqrt{N-1}} \right]\mathclose{} \braket{m}{\psi(\tau)} + \sqrt{\frac{N}{N-1}} \beta(s(\tau)) \braket{\phi}{\psi(\tau)} \\
    & = \left[ \alpha(s(\tau)) - \frac{\beta(s(\tau))}{\sqrt{N-1}} \right]\mathclose{} \psi_{m}(\tau) + \frac{\beta(s(\tau))}{\sqrt{N-1}} \bar{\psi}(\tau) .
\end{split}
\end{equation}
Note that the value of this amplitude at dimensionless time $\tau$ depends on the choice of schedule function $s$ (see \cref{tab:schedule_comparison}), which appears both in the Hamiltonian used to solve for $\chi$ and arguments of the coefficients $\alpha$ and $\beta$ for its ground state.

\printbibliography

\end{document}